\documentclass[a4paper, 11pt, oneside]{article}
\pdfoutput=1
\usepackage[a4paper, top=30mm, bottom=30mm, left=30mm, right=30mm]{geometry}
\usepackage{setspace} 
\usepackage[utf8]{inputenc}
\usepackage[T1]{fontenc}
\usepackage{lmodern} 
\usepackage[english]{babel}
\usepackage[hyphens]{url}
\usepackage[pdftex, 
            , pdfauthor={Pietro Coretto and Christian Hennig},
            , pdftitle={CONSISTENCY, BREAKDOWN ROBUSTNESS, AND ALGORITHMS FOR ROBUST IMPROPER MAXIMUM LIKELIHOOD CLUSTERING},
            , pdfsubject={Scientific article available at  arXiv:1309.6895},
            , pdfkeywords={Cluster analysis, robustness, improper density, mixture models, 
                           model-based clustering, maximum likelihood, EM-algorithm},
            , pdfproducer={},
            , pdfcreator={PdfTeX; Emacs; Linux OS}
            , colorlinks=true
            , citecolor=blue
            , linkcolor=blue
            , urlcolor=blue
            , breaklinks=true
            , pdftex=true]{hyperref}

\usepackage{etoolbox}
\usepackage{amsmath, amsthm, amssymb, amsfonts}
\usepackage{xcolor}
\usepackage{natbib} 
\usepackage{graphicx}
\usepackage{subfig}
\usepackage{caption}
\usepackage{amsmath, amsthm, amsfonts, amssymb}
\usepackage{float}
\usepackage{booktabs}
\usepackage{datetime}
\usepackage[ruled,vlined]{algorithm2e}

\definecolor{NoteColor}{RGB}{250,0,0} 

\theoremstyle{plain}
\newtheorem{theorem}{Theorem}
\newtheorem{proposition}{Proposition}
\newtheorem{lemma}{Lemma}
\newtheorem{corollary}{Corollary}
\theoremstyle{definition}
\newtheorem{definition}{Definition}

\newtheorem{remark}{Remark}

\newcommand{\conv}{\to}
\newcommand{\set}[1]{\left\{ #1 \right\}}%
\newcommand{\eps}{\varepsilon}%
\newcommand{\abs}[1]{\left| #1 \right| }%
\newcommand{\norm}[1]{\| #1 \|}%

\def\argmin{\mathop{\rm arg\,min}\limits}
\def\argmax{\mathop{\rm arg\,max}\limits}
\newcommand{\R}{\mathbb{R}}
\newcommand{\N}{\mathbb{N}}
\newcommand{\Cov}{{\rm Cov}}

\newcommand{\E}{{\rm E}}
\DeclareMathOperator{\diag}{diag}
\DeclareMathOperator{\spec}{Spec}
\DeclareMathOperator{\tr}{tr}
\DeclareMathOperator{\vect}{vect}
\newcommand{\pimax}{\pi_{\max}}
\newcommand{\lmax}{\lambda_{\max}}
\newcommand{\lmin}{\lambda_{\min}}
\newcommand{\js}{{j=1,\allowbreak  2, \allowbreak \ldots, \allowbreak G}}
\newcommand{\ks}{{k=1,\allowbreak  2, \allowbreak \ldots, \allowbreak p}}
\newcommand{\is}{{i=1,\allowbreak  2, \allowbreak \ldots, \allowbreak n}}

\newcommand{\sumin}{\sum_{i=1}^n}
\newcommand{\sumjG}{\sum_{j=1}^G}
\newcommand{\sold}{{(s)}}
\newcommand{\snew}{{(s+1)}}
\newcommand{\toldj}{\tau_{i,j}^\sold}

\newcommand{\Toldj}{T_{j}^\sold}

\newcommand{\Toldo}{T_{0}^\sold}

\newcommand{\emconv}{\text{em}}

\makeatletter
\newcommand{\mathleft}{\@fleqntrue\@mathmargin\parindent}
\newcommand{\mathcenter}{\@fleqnfalse}
\makeatother

\SetKwInOut{Data}{data}
\SetKwInOut{Parameter}{parameter}
\SetKwInOut{Input}{input}
\SetKwInOut{Output}{output}

\newcommand{\pkg}[1]{{\textsf{#1}}}

\definecolor{ChangeColor}{RGB}{117, 0, 0} 
\newcommand{\CHANGE}[1]{{  #1 } }

\title{\Large   CONSISTENCY, BREAKDOWN ROBUSTNESS,  AND ALGORITHMS 
               FOR ROBUST IMPROPER MAXIMUM LIKELIHOOD CLUSTERING}
  \author{Pietro Coretto
          \thanks{The author gratefully acknowledges financial support from  MIUR research grant PRIN 2010J3LZEN, and  the use of the high-performance computing infrastructure funded by University of Salerno (research program  ASSA098434).}\hspace{.2cm}\\ 
           Department of Economics and Statistics\\ University of Salerno, Italy \\
           \and
           Christian Hennig%
	   \thanks{The author gratefully acknowledges support from the  EPSRC grant EP/K033972/1.}\hspace{.2cm}\\
           Department of Statistical Science \\ University College London, UK}
\date{~~}

\usepackage{mdframed}
\newmdenv[
  backgroundcolor=gray!20,
  frametitle=,
  skipabove=\topsep,
  skipbelow=\topsep,
]{reminder}

\begin{document}

\begin{reminder}
{\centering \color{red} \textbf{\textsf{This is a preprint. The revised version of this paper is published as}}\\}
\vspace{5pt}
P. Coretto and C. Hennig (2017). %
``Consistency, breakdown robustness, and algorithms for robust improper maximum likelihood clustering''. %
\textit{Journal of Machine Learning Research}, Vol. 18(142), pp. 1-39 %
(\href{http://jmlr.org/papers/v18/16-382.html}{\sf download link}).
\end{reminder}

{\let\newpage\relax\maketitle}


\begin{quote}{
\small {\bf Abstract.}~The robust improper maximum  likelihood estimator (RIMLE)  is a new method for robust multivariate clustering finding approximately Gaussian clusters. It maximizes a pseudo-likelihood defined by adding a component with improper constant density for accommodating outliers to a Gaussian mixture. A special case of the RIMLE is MLE for multivariate finite Gaussian  mixture models. In this paper we treat  existence, consistency, and  breakdown theory for the RIMLE comprehensively. RIMLE's existence is proved under non-smooth covariance matrix constraints. It is shown that these can be implemented via a computationally feasible Expectation-Conditional Maximization algorithm.  \\ 


}
\end{quote}

\section{Introduction}\label{sec_introduction}

Maximum likelihood estimation (MLE) in a Gaussian mixture model with  mixture components interpreted as clusters
is a popular approach to cluster analysis (see, e.g.,  \cite{Fraley_Raftery_2002}). In many datasets not all observations can be 
assigned appropriately to clusters that can be properly modelled by a Gaussian distribution, and it is also well known that the MLE can be strongly affected by outliers  (\cite{Hennig_2004}). In this paper we investigate the recently introduced ``robust improper maximum likelihood estimator'' \citep[RIMLE, see][]{Coretto_Hennig_2014_comparison}, a method for robust clustering with clusters that can be approximated by multivariate Gaussian distributions. 
The basic idea of RIMLE is to fit an improper density to the data that is made up by a Gaussian mixture density and a ``pseudo mixture component''  defined by a small constant density, which is meant to capture outliers and observations in low density areas of the data that cannot properly be assigned to a Gaussian mixture component (called ``noise'' in the following). This is inspired by the addition of a uniform ``noise component'' to a Gaussian mixture suggested by \cite{Banfield_Raftery_1993}. \cite{Hennig_2004} showed that using an improper density improves the breakdown robustness of this approach for one-dimensional datasets.   
As in many other statistical problems, violations of the model assumptions may cause problems in cluster analysis. Our general attitude to the use of statistical models in cluster analysis is that the  models should not be understood as reflecting some underlying but in practice unobservable ``truth'', but rather as thought constructs implying a certain behaviour of methods derived from them (e.g., maximizing the likelihood), which may or may not be appropriate in a given application (more details on the general philosophy of clustering can be found in  \cite{Hennig_Liao_2013, Hennig_2015}). Using a model such as a mixture of multivariate Gaussian distributions, interpreting every mixture component as a ``cluster'', implies that we look for clusters that are approximately ``Gaussian-shaped'', but we do not want to rely on whether the data really were generated i.i.d. by  a Gaussian mixture. We focus on situations in which the  number of clusters $G$ is fixed. 

There is a number of proposals already in the literature for accounting for the presence of noise and outliers in model-based clustering problems. The contributions can be divided in two groups: methods based on mixture modelling, and methods based on fixed partition models. Within the first group \cite{Banfield_Raftery_1993} and \cite{Coretto_Hennig_2011} dealt with uniform distributions added  as ``noise components''  to a finite Gaussian mixture. \cite{McLachlan_Peel_2000b} proposed to model data based on Student $t$-distributions.  \cite{Cuesta_Albertos_etal_1997}  and \cite{GarciaEscudero_Gordaliza_1999} introduced and studied trimming in order to robustify the $k$-means partitioning method. Robust partitioning methods with homoscedastic clusters based on ML--type procedures where proposed in \cite{Gallegos_2002} and \cite{Gallegos_Ritter_2005}. Heteroscedasticity in ML-type partitioning  methods has been introduced with the TCLUST algorithm of  \cite{GarciaEscudero_etal_2008}  and the ``k--parameters clustering'' of \cite{Gallegos_Ritter_2013}. More references and an in-depth overview are giuven in \cite{GarciaEscudero_etal_2015}. Different from the methods based on fixed partition models, mixture models and RIMLE allow a smooth transition between different clusters and between clustered observations and noise, which improves parameter estimation in the presence of overlap between mixture components. The one-dimensional version of the RIMLE was  introduced in \cite{Coretto_Hennig_2010} and was investigated based on Monte Carlo experiments. Extension of the methods to the multivariate setting is not straightforward. Existence and consistency of the MLE for the multivariate Gaussian mixtures is a long standing problem due to the ill-posedness of the likelihood function. Even for ML for plain multivariate Gaussian mixtures (i.e. the RIMLE with the improper constant density set to zero), consistency theory is limited to the situation in which the model is assumed to hold precisely, and restrictive conditions are required (e.g., \cite{Redner_Walker_1984}).  \cite{Chen_Tan_2009} and \cite{Alexandrovich_2014} propose and study a penalized ML estimator. \cite{GarciaEscudero_etal_2014} studied a classification ML estimator for Gaussian mixture that is based on the TCLUST idea.  

In this paper we study the theoretical properties of the RIMLE as well as its computation.  A comprehensive treatment of existence, consistency and robustness is given. This treatment includes the case of ML for multivariate Gaussian mixture as special case.  Particularly, the robustness properties of RIMLE are superior to those of the mixture-based methods proposed by \cite{Banfield_Raftery_1993} and \cite{McLachlan_Peel_2000b}, as demonstrated later in the paper.
For fitting plain Gaussian mixtures, some issues that are treated here arise as well, particularly the need to constrain the covariance matrices in order to avoid degeneration of the likelihood. Some literature on this is cited in Section \ref{subsec_rimle_constraints}. The consistency results given here in Section \ref{sec_existence_consistency} are of a nonparametric nature and show the consistency of the RIMLE for the RIMLE-functional defined for a general class of sampling distributions. Similar results have been shown for a partition likelihood model (\cite{Gallegos_Ritter_2013}) and for alternative, trimming-based approaches to robust clustering (e.g., \cite{GarciaEscudero_etal_2008,Gallegos_Ritter_2009}). Compared to these results, there is an additional difficulty for the RIMLE, namely that degeneration of the likelihood needs to be prevented also in the case that almost all observations are assigned to the noise component and the remaining observations are fitted arbitrarily well. This may look like a disadvantage, but in the literature cited above such problems are only avoided by fixing the trimming rate.  An analysis like the one given here, and in 
\cite{Coretto_Hennig_2014_comparison}, is required for understanding the case in which both the proportion of points considered as ``noise'' and the density level at which this happens are flexible. \cite{Coretto_Hennig_2014_comparison} introduce the OTRIMLE, a  data-adaptive choice of the improper constant density, the method's tuning constant for achieving robustness. That paper also includes a comprehensive simulation study comparing the different approaches to robust clustering.   In the study, every method turns out to be superior for one or more setups, but the OTRIMLE achieves the most satisfactory overall performance. 

The paper is organized as follows. We first discuss in Section \ref{sec_data} an artificial dataset to illustrate benefits of RIMLE compared with some existing  clustering approaches. The RIMLE is introduced and defined in Section \ref{sec_rimle}. In Section \ref{sec_existence_consistency} existence and consistency of the RIMLE are proved. Section \ref{sec_practical} treats the computation of the RIMLE and the choice of input parameters for the algorithms.  Section \ref{sec_breakdown} studies the breakdown robustness of RIMLE. \CHANGE{Numerical experiments are presented in Section \ref{sec:experiment}.} Section  \ref{sec_concluding_remarks} concludes the paper.


\CHANGE{

\section{Artificial data examples} \label{sec_data}
In order to demonstrate some issues in robust clustering, we generated two artificial data sets in dimension $p=20$ from two sampling designs, called AsyNoise and GEM respectively,  also considered for the numerical experiments presented in Section \ref{sec:experiment}. The two data sets are shown in Figure \ref{fig:dataset_asynoise} and \ref{fig:dataset_gem}. A detailed description of the sampling designs is given in Section \ref{sec:experiment}. These data sets cause trouble to most clustering methods including those which account for noise. Here we discuss the results produced by appropriate clustering methods. Every method requires tuning, the choice of these tunings is extensively discussed in Section \ref{sec_choice_of_delta} and \ref{sec:experiment}. All methods treated in this section are implemented as explained in Section \ref{sec:experiment}.   
\begin{figure}[!t]
\centering
  \includegraphics[height=0.5\textheight]{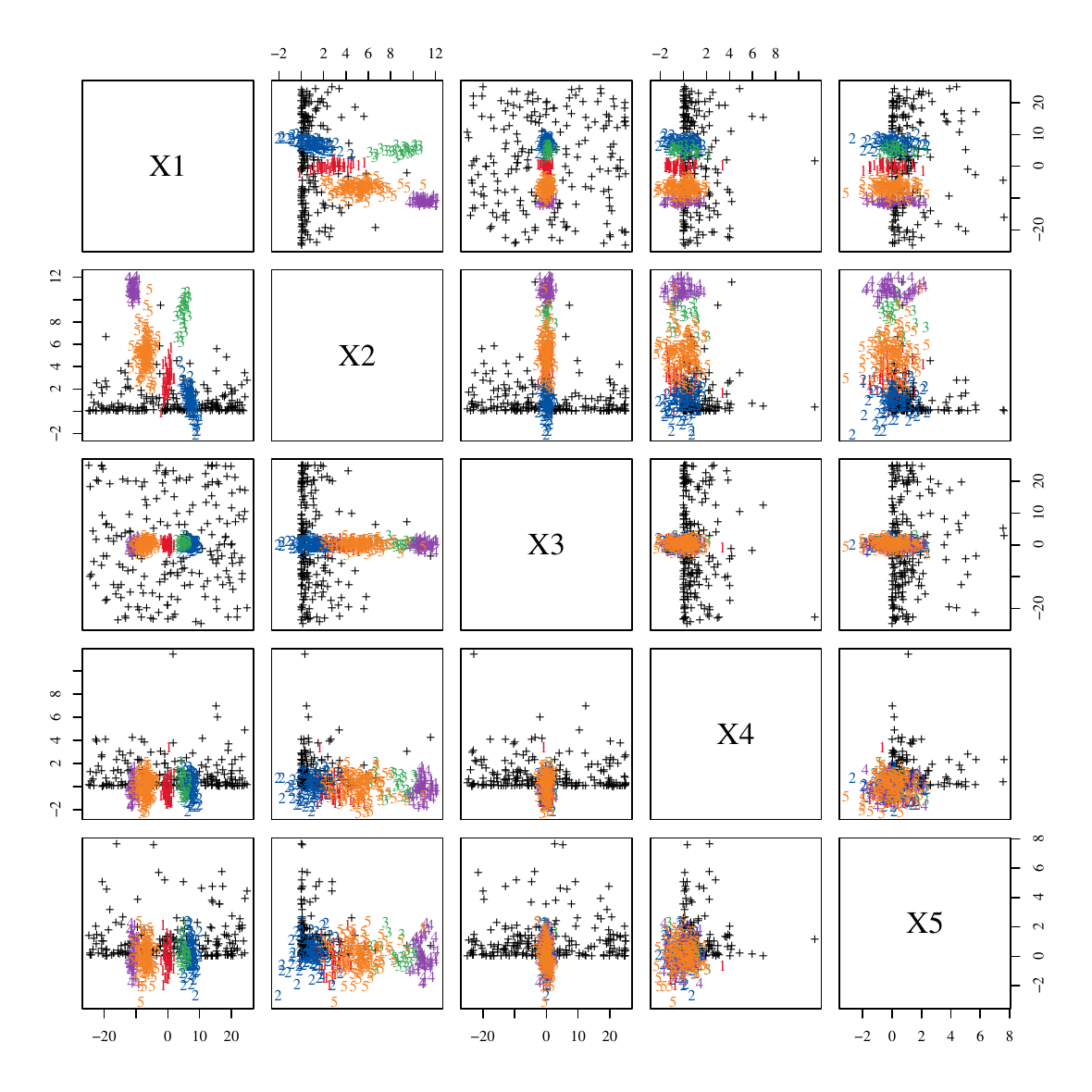}
  \caption{Scatter plots of $n=500$ data points sampled from the AsyNoise design defined in Section \ref{sec:experiment}. Marginals 1 to 5 are represented, further dimensions show a  similar pattern. Colors denote the 5 clusters, while noise is represented by the ``$+$'' symbol.}
  \label{fig:dataset_asynoise}
\end{figure}
\begin{figure}[!t]
\centering
  \includegraphics[width=0.8\linewidth]{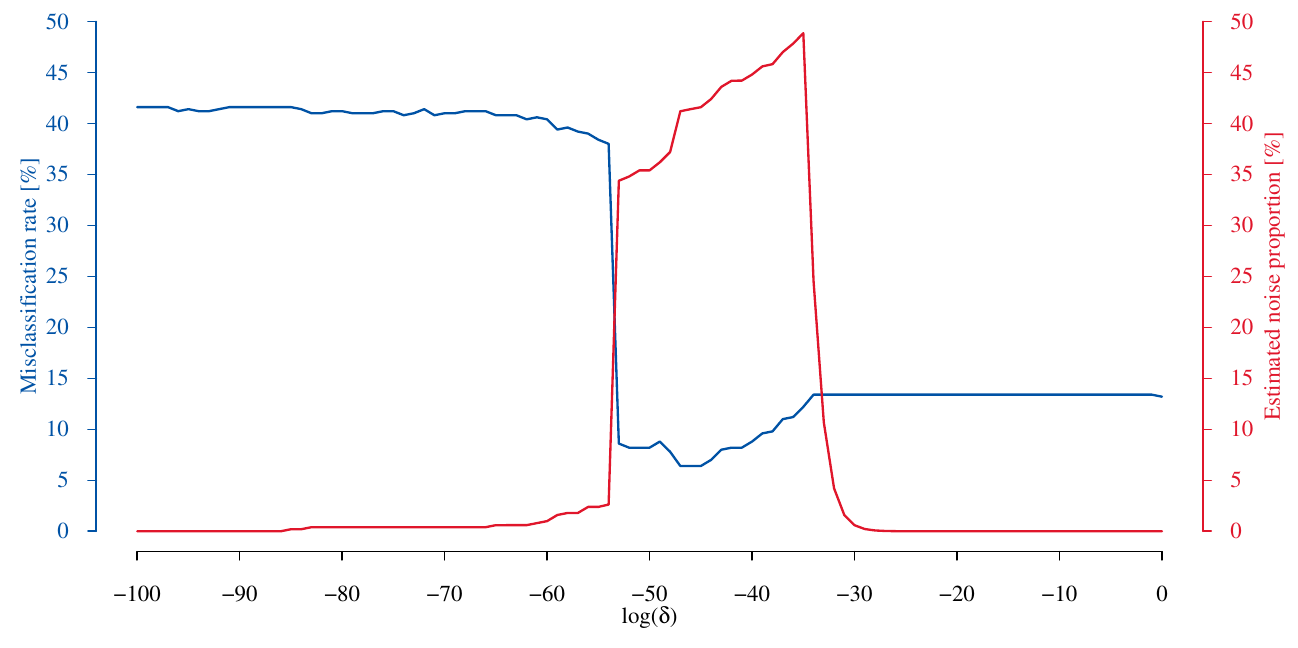}
  \caption{Misclassification rates, and estimated noise proportion by the RIMLE over a grid of $\log(\delta)$ values for the data set in Figure \ref{fig:dataset_asynoise}.}
  \label{fig:rimle_profiling_asynoise}
\end{figure}

In AsyNoise (Figure \ref{fig:dataset_asynoise}) there are 500 observations in 5 moderately separated clusters from student-t distributions with varying degrees of freedom, 187 observations (37.4\%) are background noise. Plain Gaussian mixture clustering without noise component fixing the number of clusters $G$ at 5 (as was done for all methods here) using the popular \pkg{R} package \pkg{mclust} of \cite{Mclust_Software_Manual_2012} puts the clustered points into two big clusters and assigns the noise to the remaining clusters achieving a misclassification rate of 61.4\% (i.e., the best  misclassification rate that can be achieved by permutation of the cluster  labels so that no cluster is identified with the noise). 
One could wonder whether the data set may be an easy job for clustering methods that take into account outliers, but this is not the case. The \pkg{mclust} software allows for a noise component as defined in \cite{Banfield_Raftery_1993}, which needs to be initialized. In all situations, \pkg{mclust} chooses an optimal covariance matrix parameterization based on data.  The dimensionality of the noise in the given data set seems to be too high for \pkg{mclust} to figure out that all of this really is noise. The resulting misclassification rate is 38\%. The TCLUST method implemented in the  \pkg{tclust} software of  \cite{Fritz_etal_2012} with the eigenratio constraint (a tuning parameter, see Section \ref{subsec_rimle_constraints})  set to 100 and the trimming rate set to 33\% produces a misclassification rate of 49\%. The trimming rate was fixed here imagining that one knows the true underlying expected noise proportion (which is 33\% for AsyNoise), which should normally be (unrealistically) advantageous for TCLUST. None of the five underlying clusters is found. 
\pkg{tclust}  also includes the ``ctlcurves''-tool  for guiding the user toward the choice of a suitable trimming rate. Unfortunately, this graphical tool does not give any clear indication for this data set. The true clusters can be characterized by having a considerably higher density than the noise region, so density based clustering would seem to be another promising approach, but it suffers from the high dimensionality of the data, too. We applied the DBSCAN algorithm (\cite{EsKrSaXu_1996}) using the \pkg{dbscan} \pkg{R}-package of \cite{R_dbscan}. We used various values for its tuning parameters \texttt{eps} (neighborhood radius) and \texttt{MinPts} (minimum number of points in the neighborhood). What happens for reasonable values of \texttt{MinPts} (e.g., 5) is that for small \texttt{eps} everything is classified as noise, and for large \texttt{eps} there are only 2 clusters.  In a certain short \texttt{eps}-interval more clusters are found and certain values even deliver 5  clusters. The best partition is found with \texttt{eps}=$4.25$. In this case, DBSCAN's misclassification rate is 26\%; 63\% of points are assigned to the noise. The RIMLE method treated in this paper is more appropriate. We computed it for several values of $\log(\delta)$ ($\delta$ is the value of the improper noise density; other constants were chosen as  $\gamma=100$ and $\pimax=0.5$, see Algorithm  \ref{algo:ecm} and Sections \ref{sec_rimle} for details).  Results are shown in Figure \ref{fig:rimle_profiling_asynoise}. When choosing $\log(\delta)$ appropriately, namely $\log(\delta) \in[-53,-36]$, the RIMLE gets the structure of the data set right, and it stably produces a misclassification rate of in the range [6.4\%, 11\%] with an estimated noise proportion in the range $[34.38\%, 45.8\%]$.  Values of $\log(\delta)$ below -100 do not change the results.  For large values of $\log(\delta)$ too much noise is found, hence the RIMLE's noise proportion constraint (see Section \ref{sec_rimle}) becomes active and the resulting estimated noise proportion gets close to zero.   The OTRIMLE criterion of  \cite{Coretto_Hennig_2014_comparison} selects an optimal value  $\log(\delta)=-40$, which is in the region where RIMLE shows its optimal performance. The RIMLE at $\log(\delta)=-40$ produces a misclassification rate of 8.8\% with estimated noise proportion equal to 44.8\%. 
\begin{figure}[!t]
\centering
  \includegraphics[height=0.5\textheight]{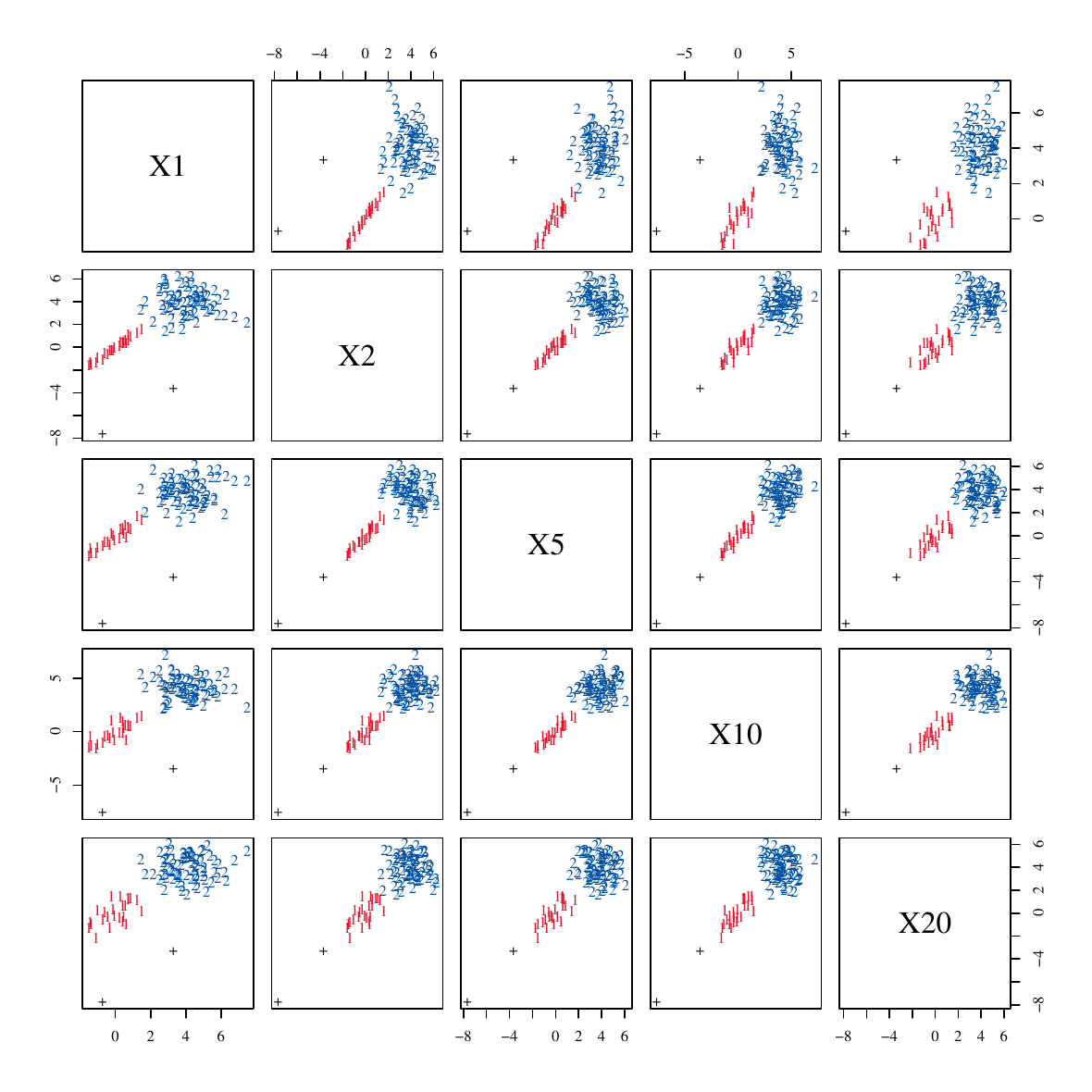}
  \caption{Scatter plots of $n=100$ points  sampled from the GEM sampling design defined in Section \ref{sec:experiment}. Marginals 1,2,5,10, and 20 are represented, further dimensions show a  similar pattern. Colors denote the 2 clusters, while noise is represented by the ``$+$'' symbol.}
  \label{fig:dataset_gem}
\end{figure}
\begin{figure}[!t]
\centering
  \includegraphics[width=0.8\linewidth]{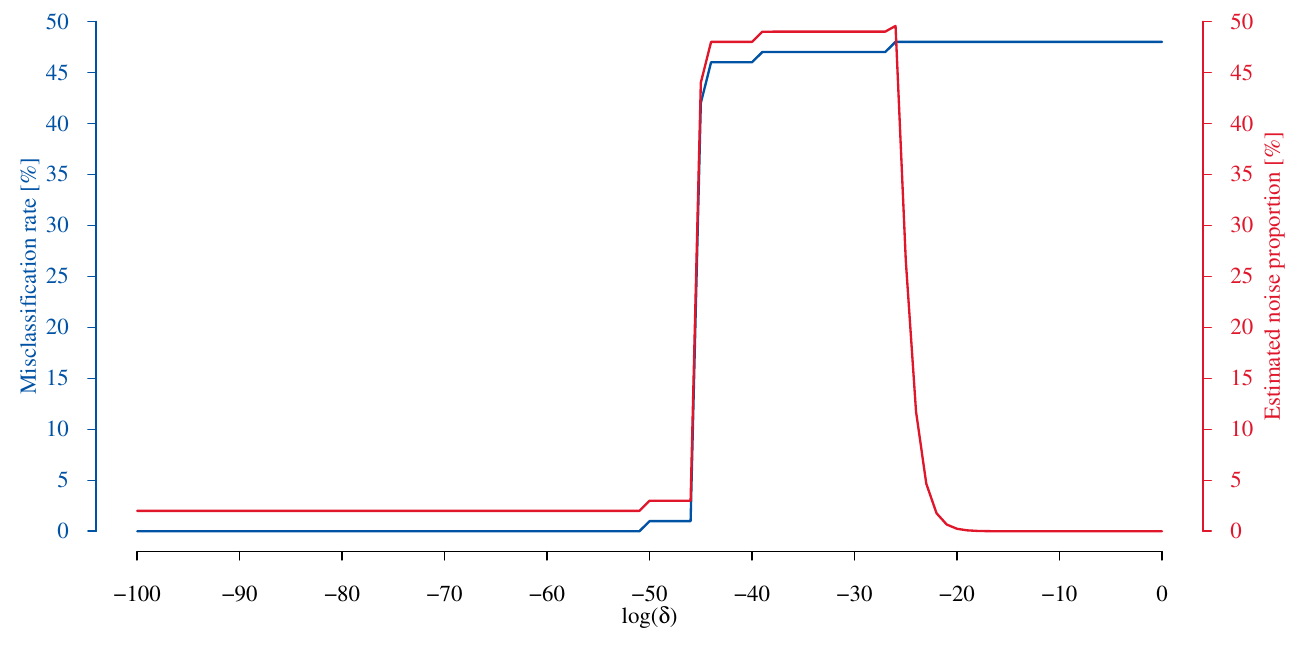}
  \caption{Misclassification rates, and estimated noise proportion by the RIMLE over a grid of $\log(\delta)$ values for the data set in Figure \ref{fig:dataset_gem}. }
  \label{fig:rimle_profiling_gem}
\end{figure}

In GEM  (Figure \ref{fig:dataset_gem}), 100 points are sampled from two normal populations with extremely different scatters, 2 points (2\%) are outliers almost lying on a hyperplane. These outliers are not extremely separated from the regular points, and this can cause trouble to robust methods. \pkg{mclust} without the noise component assigns the two outliers to cluster 1 and achieves a misclassification rate of 2\%, although the estimated mixture parameters are enormously biased. \pkg{mclust} with noise assigns most of the points of the second cluster to the noise component and the resulting misclassification rate is 47\%. In this case, the ctlcurves-tool suggests a trimming rate of 6\% for TCLUST, but the corresponding misclassification rate is 15\%. When the trimming rate is set to 2\% (which is the true expected noise proportion for the GEM design), TCLUST misclassifies 11\% of the observations. This is because  TCLUST gets the covariance structure of the second cluster wrong (see Section \ref{sec:experiment} for more details). DBSCAN does a better job with this data set, but its performance depends crucially on the appropriate setting of its two tunings. With \texttt{minPts}=3 and \texttt{eps}$\in [6.2,9.6]$, DBSCAN can reconstruct the two clusters and the noise correctly. However, in real situations, the question is how to set these parameters to achieve the best solution. As for the previous data set the RIMLE has been computed for several values of $\log(\delta)$ maintaining all other parameters as before. The result can be seen in Figure \ref{fig:rimle_profiling_gem}. For any $-\infty < \log(\delta) \leq -46$ the RIMLE is 100\% accurate and estimates a noise proportion of 2\%. The OTRIMLE method for the data-driven choice of $\delta$ selects $\log(\delta)=-200$, which is in the region where the RIMLE achieves the best results.

}


\section{Basic definitions}\label{sec_rimle}

\subsection{RIMLE and clustering}\label{subsec_rimle_clustering}
The robust improper maximum likelihood estimator (RIMLE) is based on the ``noise component''-idea for robustification of the MLE based on the Gaussian mixture model. This models the noise  by a uniform distribution, but in fact we are interested in more general  patterns of noise or outliers. However, regions of high density  are rather associated with clusters than with noise, so  the noise regions should be those with the lowest density. This kind of  distinction can be achieved by using the uniform density as in \cite{Banfield_Raftery_1993}, but in the presence of gross outliers the dependence of the uniform distribution on the convex hull of the data causes a robustness problem (\cite{Hennig_2004}).  The uniform distribution is not really used here as a model for the noise,  but rather as a technical device to account for whatever goes on in low density regions. The RIMLE drives this idea further by using an improper  uniform distribution the density value of which does not depend on how far awy extreme points in the data are from the main bulk. In the following, assume an observed sample  $x_1, x_2, \ldots, x_n,$
where $x_i$ is the realization of a random variable $X_i \in \R^p$ with $p > 1$; $X_1,\ldots,X_n$ i.i.d. The goal is to cluster the sample points into $G$ distinct groups.  RIMLE then maximizes a pseudo-likelihood, which is based on the improper pseudo-density
\begin{equation}\label{eq:psi}
\psi_\delta(x, \theta)=\pi_0 \delta + \sum_{j=1}^{G} \pi_j \phi(x; \mu_j, \Sigma_j),
\end{equation}
where $\phi(\cdot, \mu, \Sigma)$ is the Gaussian density with mean $\mu$ and covariance matrix $\Sigma$, $\pi_0,\pi_j \in [0,1]$ for $\js$, $\pi_0 + \sum\nolimits_{i=1}^G \pi_j = 1$, while $\delta$ is the improper constant density. The parameter vector $\theta$ contains  all Gaussian parameters plus all proportion parameters including $\pi_0$, ie. $\theta = (\mu_{1},\ldots,\mu_{G} ,\vect(\Sigma_{1}),\ldots,\vect(\Sigma_{G}), \pi_{0},\ldots,\pi_{G}),$ where $\vect(A)$ is the vectorized upper (or lower) triangle including the main diagonal of the symmetric  square matrix $A$. $\delta$  and the number of Gaussian components $G$ are considered fixed and known.  Although this does not define a proper probability model, it yields a useful procedure for data modelled as a proportion of $(1-\pi_0)$ of a mixture of Gaussian distributions, which have high enough density peaks to be interpreted as clusters plus a proportion 
$\pi_0$ times something unspecified with density $\le \delta$ (which may even contain further Gaussian components with so few points and/or so large within-component variation that they are not considered as ``clusters''). 

The definition of the pseudo-model in  \eqref{eq:psi} requires that the value of $\delta$ is fixed in advance. The choice of $\delta$ will be discussed in Section \ref{sec_choice_of_delta}. 

Given  the sample improper pseudo-log-likelihood function
\begin{equation}\label{eq:ln}
l_n(\theta) = \frac{1}{n} \sum_{i=1}^{n} \log \psi_\delta(x_i, \theta),
\end{equation}
the RIMLE is defined as 
\begin{equation}\label{eq:rimle}
\theta_n(\delta) =   \argmax_{\theta \in \Theta_n} l_{n}(\theta),
\end{equation}
where $\Theta_n$ is a constrained parameter space defined in Section \ref{subsec_rimle_constraints}.   $\theta_n(\delta)$ is then used to cluster points using pseudo posterior probabilities for belonging to the Gaussian components or the improper uniform. These pseudo posterior probabilities are given by 
\begin{equation*}
\tau_{j}(x_i, \theta):= \\
\begin{cases}
  \frac{\pi_0 \delta}{\psi_\delta(x_i,\theta)}                    &   \text{if} \; j=0  \\
  \frac{\pi_j \phi(x_i,\mu_j,\Sigma_j)}{\psi_\delta(x_i,\theta)}   &   \text{if} \; \js;  \\
\end{cases}
\quad \text{for} \; i=1,2,\ldots,n.
\end{equation*}
Points are assigned to the component for which the pseudo posterior probability is maximized. The assignment rule is then given by  
\begin{equation}\label{eq_J}
J(x_i, \theta):=\argmax_{j\in \set{0,1,2,\ldots,G}} \tau_{j}(x_i, \theta).
\end{equation}
The assignment based on maximum posterior probabilities  is common to all model-based clustering methods. Here, an improper density is involved, and so these are ``pseudo posterior probabilities''. 

\CHANGE{ We also define a population version of the RIMLE for later deriving consistency results for the sequence  $\set{\theta_n(\delta)}_{n \in \N}$. Let $\E_P f(x)$ be the expectation of $f(x)$ under $P$. The RIMLE population target function and the constrained parameter set can be obtained by replacing the empirical measure with $P$, and the population version of 
$l_n({\theta})$ is given by  
\begin{equation*}
L(\theta)=\E_P \log(\psi_\delta(x,\theta)).
\end{equation*}
Define  $L_G = \sup_{\theta\in\Theta_G(P)} L(\theta)$, where $\Theta_G(P)$ is a constrained parameter space defined in Section \ref{subsec_rimle_constraints}.}

\subsection{The constrained parameter space}\label{subsec_rimle_constraints}

Some notation: the $k$th element of $\mu_j$ is denoted by $\mu_{j,k}$ for $\ks$ and $\js$. Let $\lambda_{j,k}$ be the $k$th  eigenvalue of $\Sigma_j$,  define  $\Lambda(\theta) =\set{\lambda_{j,k}: \; \js; \; k=1,2, \ldots,p }$, $\lmin(\theta) = \min_{j,k}\{\Lambda(\theta)\}$,  $\lmax(\theta)=\max_{j,k}\{\Lambda(\theta)\}$. 
\begin{remark}\label{remark:phi}
 The $p$-dimensional Gaussian density can be written in terms of the eigen-de\-com\-po\-si\-tion of the covariance matrix:
\begin{equation*}
\phi(x; \mu, \Sigma) = (2\pi)^{-\frac{p}{2}} \left( \prod_{k=1}^p \lambda_k \right)^{-\frac{1}{2}} \exp{\left(  -\frac{1}{2} \sum_{k=1}^p  \lambda_k^{-1} (x-\mu)' v_kv_k' (x-\mu) \right)},
\end{equation*}
where $\lambda_k$ is the $k$-th eigenvalue of $\Sigma$, and $v_k$ is its associated eigenvector, for $\ks$. Let $\lambda_0=\min\{\lambda_k; \; \ks\}$. Then,
$\lim_{\lambda_0 \searrow 0} \phi(\mu; \mu,\Sigma) = \infty$, with $\phi(\mu; \mu,\Sigma) = O(\lambda_0^{-p/2})$ as $\lambda_0 \searrow 0$. On the other hand  $\lim_{\lambda_0 \searrow 0} \phi(x; \mu,\Sigma) = 0$ for all $x \neq \mu$, with $\phi(x; \mu,\Sigma) = o(\lambda_0^q)$ for any fixed $q$ as $\lambda_0 \searrow 0$. This implies that 
\begin{equation*}
\lim_{\lambda_0 \searrow 0} \; \phi(\mu; \mu,\Sigma)\phi(x; \mu,\Sigma) \to 0 \quad \text{for any} \;x \neq \mu
\end{equation*}
Furthermore, each of the  density components in $\psi_\delta(\cdot)$ can be bounded above in terms of $\lmax(\theta)$ and $\lmin(\theta)$:
\begin{equation}\label{eq:phi_max}
\phi(x; \mu_j, \Sigma_j) \leq (2\pi\lmin(\theta))^{-\frac{p}{2}}  \exp \left\{ -\frac{1}{2} \lmax(\theta)^{-1}\norm{x-\mu_j}^2 \right\}  \leq (2\pi\lmin(\theta))^{-\frac{p}{2}}.
\end{equation}
\end{remark}
The optimization  problem in \eqref{eq:rimle} requires that $\Theta_n$ is suitably defined, otherwise $\theta_n(\delta)$ may not exist. Consider a sequence $(\theta_m)_{m\in\N}$,  as discovered by \cite{Kiefer_Wolfowitz_1956}, the likelihood of a Gaussian mixtures degenerates if $\lambda_{1,1,m} \searrow 0$ if $\mu_{1,m}=x_1$, and this holds for \eqref{eq:ln}, too.
We here use the eigenratio constraint 
\begin{equation} \label{eq_cov_constraint}
\lmax(\theta) / \lmin(\theta)  \leq \gamma < +\infty
\end{equation}
with a constant $\gamma\ge 1$, where $\gamma=1$ constrains all component covariance matrices to be spherical and equal, as in $k$-means clustering.  This type of constraint has been proposed by \cite{Dennis_1981}, while \cite{Hathaway_1985} showed consistency of the scale-ratio constrained MLE for one-dimensional Gaussian mixtures. \cite{Ingrassia_2004} and \cite{Ingrassia_Rocci_2007} introduced EM algorithms for implementing these constraints for multivariate datasets. TCLUST by \cite{GarciaEscudero_etal_2008} and  \cite{GarciaEscudero_etal_2014} also makes use of eigenratio constraints. Moreover there are a number of alternative constraints, see \cite{Ingrassia_Rocci_2011, Gallegos_Ritter_2009}. It may be seen as a disadvantage of \eqref{eq_cov_constraint} that the resulting estimator will not be affine equivariant (this would require allowing $\lmax(\theta) / \lmin(\theta)\to\infty$ within any component). \CHANGE{ Affine equivariance can be achieved by defining a sphered version of the RIMLE as 
\begin{displaymath}
\theta_n^*(\delta)=\theta_n^*(\delta,x_1,\ldots,x_n)=\theta_n(\delta,x_1^*, x_2^*, \ldots, x_n^*)
\end{displaymath}
with $x_i^*=A(x_i-m),\ i=1,\dots,n$, where $S(x_1,\ldots,x_n)^{-1}=A'A$; 
$S(x_1,\ldots,x_n)$ could be the sample covariance matrix or another scale 
matrix and $m$ the mean vector or another location estimator. 
This yields affine equivariance because 
the sphered versions of $\{x_1,\ldots,x_n\}$  and 
$\{Bx_1+b,\ldots,Bx_n+b\}$ with some invertible $p\times p$-matrix $B$ 
and $b\in\R^p$ are 
the same. Affine equivariance is not necessarily desirable though, see 
\cite{Hennig2015b}, Sec. 31.3.4.}

\CHANGE{ This defines the parameter space
\begin{equation}\label{eq_Theta_tilde}
\tilde{\Theta}:= \left\{
\theta:  \; \; %
\pi_j \geq 0 \; \forall j\geq 1, \; %
\pi_0+\sum_{j=1}^G \pi_j =1; \; %
\frac{\lmax(\theta)}{\lmin(\theta)} \leq \gamma
\right\}. 
\end{equation}}
\CHANGE{Occasionally, later, the notation $\|\theta\|$ will refer to the
Euclidean norm of a vector pieced together from all the parameters
collected in $\theta$, in which all covariance matrices are interpreted as 
subvectors of all the matrix entries.} 

Although \eqref{eq_cov_constraint} ensures the boundedness of the likelihood in standard mixture models and TCLUST, for RIMLE this is not enough. \CHANGE{ 
The Gaussian components could degenerate on a few points and all other points 
could be fitted by the improper uniform component.
Therefore we impose an additional constraint:
\begin{equation} \label{eq_noise_constraint}
\frac{1}{n}  \sum_{i=1}^n  \tau_0(x_i, \theta) \leq \pi_{\text{max}},
\end{equation}
for fixed $0< \pimax < 1$. The quantity $ n^{-1} \sum_{i=1}^n  \tau_0(x_i,\theta)$ can be interpreted as the estimated proportion of noise points.
This constraint depends on the dataset. Unfortunately the similar
looking constraint $\pi_0\leq\pimax$ independent of the data 
will not do, because this could not stop more than a portion of $\pi_0$ points
to be fitted by the improper uniform component.
  
There is therefore a constrained effective parameter space for RIMLE estimation
depending on the dataset:}
\begin{equation}\label{eq_Theta_n}
\Theta_n:= \left\{
\theta \in \tilde{\Theta}:  \; \; %
\pi_j \geq 0 \; \forall j\geq 1, \; %
\pi_0+\sum_{j=1}^G \pi_j =1; \;
\frac{1}{n}  \sum_{i=1}^n  \tau_0(x_i, \theta) \leq \pi_{\text{max}}; \; %
\frac{\lmax(\theta)}{\lmin(\theta)} \leq \gamma
\right\}. 
\end{equation}
\CHANGE{ Analogously, existence and consistency of the RIMLE functional can only be showed on a parameter subset of $\tilde{\Theta}$ that depends on the underlying distribution and enforces that enough probability mass is fitted by Gaussian components rather than the improper uniform:
\begin{equation}\label{eq_Theta_G}
\Theta_G(P):= \left\{
\theta\in \tilde{\Theta}:  \;  %
\pi_j \geq 0 \; \forall j \geq 1; \; %
\pi_0+\sum_{j=1}^G \pi_j =1; \;
 \E_P\frac{\pi_0\delta}{\psi_\delta(x,\theta)}\le\pimax; \; %
\frac{\lmax(\theta)}{\lmin(\theta)} \leq \gamma
\right\}.
\end{equation}}

\section{RIMLE existence and consistency}\label{sec_existence_consistency}

We first show existence of the RIMLE for finite samples 
Let $\#(A)$ denote the cardinality of the set $A$. 
Let $\underline{x_n}=\{x_1, x_2, \ldots, x_n\}$. 
Lemma \ref{lemma_finite_n_eigen_1} concerns the important case of plain Gaussian mixtures
($\delta=0$) and requires a weaker assumption A0(a) for existence than A0 required for the RIMLE with $\delta>0$. Here are some assumptions:
\begin{description}
\item[A0(a)] $\#(\underline{x_n})>G$.
\item[A0~~~]    $\#(\underline{x_n})>G+\lceil n \pimax  \rceil$. 
\end{description}

\begin{lemma}\label{lemma_finite_n_eigen_1}
Assume A0(a), $\delta=0$. Let  $(\theta_m)_{m \in \N}$ be a sequence such that $\lmax(\theta_m) / \lmin(\theta_m)$ $\leq\gamma$. Assume also that  for some $\js$ and $\ks$,  $\lambda_{k,j,m} \searrow 0$ as  $m \to \infty$;   then $\sup l_n(\theta_m) \to -\infty.$
\end{lemma}
\begin{proof}
$\lambda_{k,j,m} \searrow 0$ implies $\lmax(\theta_m)\searrow 0,\ 
\lmin(\theta_m)\searrow 0$ at the same speed because of
\eqref{eq_cov_constraint}. Assume w.l.o.g. (otherwise consider a suitable
subsequence) that $(\theta_m)_{m \in \N}$ is such that $\mu_{j,m},$ $\js$ either
leave every compact set for $m$ large enough or converge, and assume
w.l.o.g., that if their limits are in $\{x_1,\ldots,x_n\}$, they are in  
$\underline{x_G}=\{x_1,\ldots,x_{G}\}$. A0(a) implies that 
$\exists\ x_i\not\in\underline{x_G}$, and $\exists \nu>0$ such that for all such $x_i,\ \js$ and large enough $m: \|x_i-\mu_{j,m}\|\ge \nu$. Because the likelihood
\begin{equation}\label{eq_prod_delta_0}
\mathcal{L}_{n}(\theta_m) =  %
\prod_{x_i\in\underline{x_G}} \Bigl\{\sum_{j=1}^{G}  \pi_{j}  \phi(x_{i}; \mu_{j,m}, \Sigma_{j,m}) \Bigr\} %
\prod_{x_i\not\in\underline{x_G}} \Bigl\{\sum_{j=1}^{G} \pi_{j} \phi(x_{i}; \mu_{j,m}, \Sigma_{j,m})  \Bigr\}, %
\end{equation}
and Remark  \ref{remark:phi}, the first product is of order $O(\lmin(\theta_m)^{-p/2})^G$, and the second one of order $o(\lmin(\theta_m)^q)$ for any fixed $q$, which implies 
that $\mathcal{L}_{n}(\theta_m) \to 0$ and  $l_{n}(\theta_m) \to -\infty$.
\end{proof}

\begin{lemma}\label{lemma_finite_n_eigen_2}
Assume A0, $\delta>0$.  $(\theta_m)_{m \in N}$ is a sequence in $\Theta_n$. Assume also that  for some $\js$ and $\ks$,  $\lambda_{k,j,m} \searrow 0$  as $m \to \infty$.   Then $l_n(\theta_m) \to -\infty.$ 
\end{lemma}
\begin{proof}
Using the definitions of the proof of Lemma \ref{lemma_finite_n_eigen_1}, instead of \eqref{eq_prod_delta_0} now 
\begin{equation}\label{eq_prod_delta_positive}
\mathcal{L}_{n}(\theta_m) =  %
\prod_{x_i\in\underline{x_G}} \Bigl\{ \pi_{0,m}\delta +\sum_{j=1}^{G}   \pi_{j}  \phi(x_{i}; \mu_{j,m}, \Sigma_{j,m}) \Bigr\} %
\prod_{x_i\not\in\underline{x_G}} \Bigl\{\pi_{0,m}\delta + \sum_{j=1}^{G} \pi_{j} \phi(x_{i}; \mu_{j,m}, \Sigma_{j,m})  \Bigr\} %
\end{equation}
has to be considered, so that the limit behaviour of $(\pi_{0,m})_{m \in \N}$
is relevant. 
\eqref{eq_noise_constraint} implies
\begin{equation}\label{eq_npc_small_eigenvalues}
\frac{1}{n} \sum_{x_i\in\underline{x_G}}  \left( 1 + \sum_{j=1}^G \frac{\pi_j \phi(x_i, \mu_{j,m}, \Sigma_{j,m})}{\pi_{0,m} \delta}   \right)^{-1}+ %
\frac{1}{n} \sum_{x_i\not\in\underline{x_G}}  \left( 1 + \sum_{j=1}^G \frac{\pi_j \phi(x_i, \mu_{j,m}, \Sigma_{j,m})}{\pi_{0,m} \delta}  \right)^{-1} %
\leq \pimax.
\end{equation}
Suppose that $(\pi_{0,m})_{m \in \N}$ does
not converge to zero as $O(\phi(x_i, \mu_{j,m}, \Sigma_{j,m}))$ for at least one 
$x_i\not\in\underline{x_G}$.
For  $m \to \infty$,  the left term of \eqref{eq_npc_small_eigenvalues} is $\ge 0$, and the right term (at least a subsequence) converges to $\frac{\#(\underline{x_n}\setminus \underline{x_G})}{n}$, which A0 requires to be
$>\pimax$ with contradiction, thus $\pi_{0,m}=O(\phi(x_i, x_j, \Sigma_{j,m}))$.
Therefore, by the same argument as in the proof of Lemma \ref{lemma_finite_n_eigen_1}, the right product in \eqref{eq_prod_delta_positive} vanishes fast enough so that $l_{n}(\theta_m) \to -\infty$.
\end{proof}

From these Lemmas:
\begin{theorem}[Finite Sample Existence]\label{th:finite_sample_existence}
Assume A0. Then $\theta_n(\delta)$ exists for all $\delta \geq 0$.
\end{theorem}
\begin{proof}
$\Theta_n$ depends on $\delta$ via \eqref{eq_noise_constraint}. $\Theta_n$ is not empty for any $\delta$, because for any fixed values of the other parameters, small enough $\pi_0$ will fulfil \eqref{eq_noise_constraint}.
Next show that there exists a compact set $K_n \subset \Theta_n$ such that $\sup_{\theta \in K_n} l_n(\theta) = \sup_{\theta \in \Theta_n} l_n(\theta)$. \\
{\it Step A:} consider $\theta$ such that  $\pi_1=1$, $\mu_1=x_1$, $\Sigma_j=I_p$ for all $\js$, arbitrary $\mu_j$ and $\pi_j=0$ for all $j\neq 1$.  For this,
$l_n(\theta) = \sum_{i=1}^n \log \phi(x_i; x_1,\Sigma_1) > -\infty$, thus $\sup_{\theta \in \Theta_n} l_n(\theta)>-\infty$.\\
{\it Step B:} consider a sequence $(\dot \theta_m )_{m\in N}$. It needs to be proved that if $(\dot \theta_m )_{m\in N}$ leaves a suitably chosen compact set $K_n$, it cannot achieve as large values of $l_n$ as one could find within $K_n$. 
Lemma \ref{lemma_finite_n_eigen_2} (Lemma \ref{lemma_finite_n_eigen_1} for $\delta=0$) rules out the possibility of any  $\lambda_{k,j,m} \searrow 0$.\\
{\it Step C:} \eqref{eq:phi_max} implies that $l_n(\theta)$ can be bounded from above in terms of $\pi_0, \lmin$ and $\delta$:
\begin{equation*}
l_n(\theta) =    \frac{1}{n} \sum_{i=1}^n \log \left( \pi_0 \delta + \sum_{j=1}^G \pi_j \phi(x_i; \mu_j, \Sigma_j)  \right) %
\leq   \log \left( \pi_0 \delta + (1-\pi_0)   (2\pi\lambda_{\min}(\theta))^{-\frac{p}{2}} \right). 
\end{equation*}
Consider $\dot \theta  \in \Theta_n$ such that  $\dot{\lambda}_{k,j} < +\infty$ for all $\ks$ and $\js$ (using the obvious notation of components of ``dotted''
parameter vectors). Also consider a sequence 
$(\ddot \theta_m)_{m\in \N}$  such that $\ddot \theta_m \conv \ddot \theta$ where $\ddot \theta$ is equal to $\dot \theta$ except that $\ddot{\lambda}_{k,j,m} \conv +\infty$ for some $k \in\{1,\ldots,p\}$ and $j \in\{1,\ldots,G\}$. 
By \eqref{eq_cov_constraint}, $\lambda_{\min}(\ddot \theta_m) \conv +\infty$ and thus $l_n(\ddot \theta_m) \conv \log(\ddot \pi_0\delta)$. Clearly $\ddot \pi_0<1$ because otherwise  $n^{-1}\sum_{i=1}^n  \tau_0(x_i, \ddot \theta)=1$, violating \eqref{eq_noise_constraint}. Therefore 
 $\lim_{m\conv \infty }l_n(\ddot\theta_m) \leq l_n(\dot \theta)$.\\
{\it Step D:} now consider $||\dot \mu_{j,m}|| \to +\infty$, for $j=1$, w.l.o.g. Choose $\ddot \theta_m$ equal to $\dot \theta_m$ except now $\ddot \mu_{1,m}=0$ for all $m$. Note that $\phi(x_i;\dot \mu_{1,m}, \dot \Sigma_1) \to 0$ for all $\is$, which implies that  $\psi_\delta(x_i,\dot \theta_m) < \psi_\delta(x_i,\ddot \theta_m)$ for large enough $m$, for which then $l_n(\dot \theta_m) < l_n(\ddot \theta_m)$. Applying this argument to all $j$ with  $||\dot \mu_{j,m}|| \to +\infty$ shows that better $l_n$ can be achieved inside a compact set.\\

Continuity of $l_n$ now guarantees existence of $\theta_n(\delta)$. 
\end{proof}

\CHANGE{ We now derive consistency for the sequence  
$\set{\theta_n(\delta)}_{n \in \N}$
as estimator of $L_G$.}
Consistency of the RIMLE can be achieved only if  $L_G$ exists. In order to ease the notation we define $\eta(x, \theta)= \sum_{j=1}^{G} \pi_{j} \phi(x; \mu_j, \Sigma_j).$ Consider the following assumptions on $P$:
\begin{description}
\item[A1] For every $x_1,\ldots,x_G\in \R^p: P\{x_1,\ldots,x_G\}<1-\pi_{max}$.
\item[A2] $L_G>L_{G-1}$, where for $G\in \N$ let $L_{G}=\sup_{\theta\in\Theta_G(P)}L(\theta)$.  
\item[A3] There exist $\epsilon_1,\epsilon_2>0$ so that for every  $\theta$ with $\pi_0<\epsilon_1:$ $L(\theta) \le L_G-\epsilon_2$.
\end{description}
\begin{remark}\label{remark:assumption} 
Assumption A1 requires that no set of $G$ points carries probability $1-\pimax$ or more. Otherwise the log-likelihood can be driven to $\infty$ by fitting $G$ mixture components to $G$ points with all covariance matrix  eigenvalues converging to zero. The improper noise component could take  care of all other points. Note that assumptions A2 and A3 are not both required, but only any single one of them. 
 
\CHANGE{ A2 states that $G$ mixture components fit the data better than $G-1$. If this is not the case, there is at least one redundant component, and one cannot make sure that $L(\theta)$ is bounded away from $L_G$ for large $n$ in some distance from the ``true'' RIMLE-functional as the redundant component can be moved around, see Theorem \ref{th:asy_existence}. In case that A2 is not fulfilled, a weaker result can still be achieved, namely the existence of a not necessarily unique consistent sequence of local maximizers of $l_n$. This requires A3, which states that a noise proportion bounded away from zero is required for maximizing $L$. If} neither A2 nor A3 are fulfilled, $P$ can be fitted perfectly by fewer than $G$ mixture components and no noise. In this case one cannot stop the remaining mixture components from \CHANGE{leaving every compact set, and therefore one cannot expect consistency of all components for any method; as long as there is still noise bounded away from zero, those mixture components still can contribute to fitting what otherwise would be noise, and fits become worse if these degenerate.} 

Note that this is less often the case than one
might expect; for example, a plain Gaussian mixture with $G-1$ components may still fulfill A3: if the density of one of the components is uniformly smaller than $\delta$, a better pseudo-likelihood can obviously be achieved by assigning its proportion to the noise component than by choosing $\pi_0=\pi_G=0$ and otherwise the true parameters. A Gaussian component that ``looks like noise'' rather than like a ``cluster'' will be treated as noise. 
\end{remark}
\begin{lemma}\label{lasexist} %
For any probability measure $P$ on $\R^p$,  $L_G>-\infty$.
\end{lemma}
\begin{proof}
Choose compact $K\subset \R^p$ with $P(K)>0$. Let $q=1-P(K)$ and choose $K$ big enough that $\pimax>q$.  Choose $\mu_1=\E_P(x|x\in K),\ \Sigma_1=\Cov_P(x|x\in K),\ \pi_2=\ldots=\pi_G=0$. If all eigenvalues of $\Sigma_1$ are zero, choose $\Sigma_1=I_p$. Let $\lambda_{\max,1}$ be the largest eigenvalue of $\Sigma_1$. If  $\lambda_{\max,1}/\lambda_{i,1} > \gamma$ for any eigenvalue $\lambda_{i,1}$ of  $\Sigma_1$, modify  $\Sigma_1$ by replacing all eigenvalues smaller than $\gamma\lambda_{\max,1}$ by  $\gamma\lambda_{\max,1}$ in its spectral decomposition. Let $\phi_{\min}=\min_{x\in K} \phi(x,\mu_1,\Sigma_1)>0$. Choose 
\begin{displaymath}
\pi_0=\frac{(\pimax-q)\phi_{\min}}{2((1-\pimax)\delta+(\pimax-q)\phi_{\min})} > 0,\ \pi_1=1-\pi_0.
\end{displaymath}
Observe that the resulting $\theta\in \Theta_G(P)$ (with all other parameters  chosen arbitrarily) by 
\begin{align*}
\E_P\frac{\pi_0\delta}{\psi_{\delta}(x,\theta)} %
   & =  \int_K  \frac{\pi_0\delta} {\pi_0\delta+\eta(x,\theta)} dP(x) + \int_{K^c}\frac{\pi_0\delta}{\pi_0\delta+\eta(x,\theta)} dP(x) \leq \\
   &\leq (1-q)\frac{\pi_0\delta}{\pi_0\delta+(1-\pi_0)\phi_{\min}}+q. 
\end{align*}
This is smaller than $\pimax$ if  $\pi_0<\frac{(\pimax-q)\phi_{\min}}{(1-\pimax)c+(\pimax-q)\phi_{\min}}$. Furthermore,  $L(\theta)\ge (1-q)\log(\pi_0\delta+(1-\pi_0)\phi_{\min})+q\log(\pi_0\delta)>-\infty$.
\end{proof}
\begin{lemma}\label{laslambda} Assume A1. There are $\lambda^*_{\min}>0,\ \lambda^*_{\max}<\infty,\ \epsilon>0$,  so that  
\begin{description}
\item[(a)]
$L(\theta)\le L_G-\epsilon$ for every $\theta$ with  $\lmin(\theta)<\lmin^*$ or $\lmax(\theta)>\lmax^*$,
\item[(b)] for $x_1,x_2,\ldots$ i.i.d. with ${\cal L}(x_1)=P$, 
for sequences $(\theta_n)_{n\in\N}$ with
$\lmin(\theta_n)<\lmin^*$ or $\lmax(\theta_n)>\lmax^*$ for large enough $n$: 
$l_n(\theta_n)\le l_n(\theta_n(\delta))-\epsilon$ a.s.
\end{description}
\end{lemma}
\begin{proof}
Start with part (a). First consider a sequence $\set{\theta_m}_{m\in\N}\in \Theta_G(P)^\N$ with  $\lambda_{\max}(\theta_m)\to\infty$. The eigenvalue ratio constraint forces all  covariance matrix eigenvalues to infinity, and therefore  $\sup_x \phi(x,\mu_{j,m},\Sigma_{j,m})\searrow 0$. But this means that  $\E_P\frac{\pi_0\delta}{\psi_{\delta}(x,\theta)}\to 1 >\pimax$ and $\theta_m\not\in\Theta_G(P)$ eventually, unless $\pi_{0,m}\searrow 0$, too. If the latter is the case, $\psi_{\delta}(x,\theta)\searrow 0$ uniformly over all $x$ and $L(\theta_m)\searrow -\infty$, which together with Lemma \ref{lasexist} makes it impossible that $L(\theta_m)$ is close to $L_G$ for $m$ large  enough and $\lambda_{\max}(\theta_m)$ too large, proving the existence of  the upper bound $\lambda^*_{\max}<\infty$ as required.\\
Now consider a sequence $\set{\theta_m}_{m\in\N}\in \Theta_G(P)^\N$ with  $\lambda_{\min}(\theta_m)\to 0$.  Define 
$$A_{m,\epsilon} = \set{ x: \; \min_{\js} \|x-\mu_{j,m}\| > \epsilon }.$$ 
A1 ensures that for $0<\epsilon_3$ there exists $\epsilon>0$ so that for all $m\in \N:$ $P(A_{m,\epsilon})\ge \pimax+\epsilon_3$. Based on \eqref{eq:phi_max} derive an upper bound for $\pi_{0,m}$ from the constraint $\int \frac{\pi_{0,m}\delta}{\pi_{0,m}\delta+\eta(x,\theta_m)}dP(x) \le \pimax$:
\begin{align*}
\int \frac{\pi_{0,m}\delta}{\pi_{0,m}\delta+\eta(x,\theta_m)}dP(x) 
    &=  \int_{A_{m,\epsilon}} \frac{\pi_{0,m}\delta}{\pi_{0,m}\delta+\eta(x,\theta_m)}dP(x)+ 
         \int_{A_{m,\epsilon}^c} \frac{\pi_{0,m}\delta}{\pi_{0,m}\delta+\eta(x,\theta_m)}dP(x) \geq \\
    &\geq  P(A_{m,\epsilon}) \frac{\pi_{0,m}\delta}{\pi_{0,m}\delta + \max_{x\in A_{m,\epsilon}}\eta(x,\theta_m)}, 
\end{align*}
which by \eqref{eq:phi_max}  implies
$$
\pi_{0,m} \le  \frac{\pimax \max_{x\in A_{m,\epsilon}}\eta(x,\theta_m)}{\delta(P(A_{m,\epsilon})-\pimax)} \le \frac{\pimax\exp(-\frac{\epsilon^2}{2\gamma\lmin(\theta_m)})} {\delta \epsilon_3 (2\pi)^{p/2}\lmin(\theta_m)^{p/2}}. 
$$
For the log-likelihood,
\begin{align*}
L(\theta_m) 
   & =    \int_{A_{m,\epsilon}} \log(\pi_{0,m}\delta+\eta(x,\theta_m)) dP(x) + \int_{A_{m,\epsilon}^c} \log(\pi_{0,m}\delta+\eta(x,\theta_m)) dP(x) \leq   \\
   &\leq   \int_{A_{m,\epsilon}}\log\left(\frac{\delta\pimax\exp(-\frac{\epsilon^2}{2\gamma\lmin(\theta_m)})} {\delta\epsilon_3
           (2\pi)^{p/2}\lmin(\theta_m)^{p/2}}+\frac{\exp(-\frac{\epsilon^2}{2\gamma\lmin(\theta_m)})} 
           {(2\pi)^{p/2}\lmin(\theta_m)^{p/2}}\right)dP(x) + \\
   &{ }   {~~~} + \int_{A_{m,\epsilon}^c}\log\left(\frac{\delta\pimax\exp(-\frac{\epsilon^2}{2\gamma\lmin(\theta_m)})} 
             {\delta\epsilon_3 (2\pi)^{p/2}\lmin(\theta_m)^{p/2}}+\frac{1} {(2\pi)^{p/2}\lmin(\theta_m)^{p/2}}\right)dP(x) \leq \\
   &\leq   P(A_{m,\epsilon})\left(-\frac{c_1}{\lmin(\theta_m)}- c_2\log(\lmin(\theta_m))+c_3\right) +\\
   &{}     {~~~} + P(A_{m,\epsilon}^c)\left(o(1)- c_4\log(\lmin(\theta_m))+c_5\right) = \\
   & =    -\frac{c_6}{\lmin(\theta_m)}-c_7\log(\lmin(\theta_m))+c_8
\end{align*} 
for positive constants $c_1,c_2,c_4,c_6,c_7$ and constants $c_3, c_5, c_8$,  all  independent of $\theta_m$. If $\lmin(\theta_m)\searrow 0$, this implies
$L(\theta_m)\searrow-\infty$, proving together with Lemma \ref{lasexist} the existence of  the lower bound $\lambda^*_{\min}>0$.

Part (b) holds because if $(\theta_m)_{m\in\N}$ is chosen as above for 
$m=n\to\infty$ and $P$ is replaced by the empirical distribution $P_n$, Glivenko-Cantelli enforces $P_n(A_{n,\epsilon})-P(A_{n,\epsilon})\to 0$ a.s. \CHANGE{ 
Glivenko-Cantelli applies 
because the class of all $A_{m,\epsilon}$ is a subset of the class of 
intersections of the complements of all closed balls, and therefore a
Vapnik-Chervonenkis class, see \cite{Van_der_Vaart_Wellner_1996}.}
The argument carries over using all other integrals in the finite sample-form, i.e., w.r.t. $P_n$. Lemma \ref{lasexist} carries over because $l_n(\theta)\to L(\theta)$ a.s. by the strong law of large numbers for $\theta$ with $L(\theta)>-\infty$.
\end{proof}
\begin{remark}\label{rasphimax} 
Lemma \ref{laslambda} (a) and \eqref{eq:phi_max} imply that for all $\theta$ with  $L(\theta)>L_G-\epsilon$, $\js$ and all $x$: 
\begin{equation*}
\phi(x,\mu_j,\Sigma_j) \le \phi_{\max}= \frac{1}{(2\pi)^{p/2}(\lmin^*)^{p/2}}.
\end{equation*}
This implies $L_G<\infty$.

The same holds because of Lemma \ref{laslambda} (b) 
for $x_1,x_2,\ldots$ i.i.d. with ${\cal L}(x_1)=P$ for large 
enough $n$ a.s. for all $\theta$ with 
$l_n(\theta)>l_n(\theta_n(\delta))-\epsilon$.
\end{remark}
\begin{lemma}\label{lasmua3}  
Assume A1 and A2.  There is a compact set $K\subset\R^p$  so that 
\begin{description}
\item[(a)] $L$ reaches its supremum for $\mu_1,\ldots,\mu_G\in K$ and is bounded away from the supremum 
if not all of $\mu_1,\ldots,\mu_G\in K$ 
\CHANGE{ (i.e., $\exists \epsilon_4>0$ so that 
it is $\le \sup L-\epsilon_4$)},
\item[(b)] for $x_1,x_2,\ldots$ i.i.d. with ${\cal L}(x_1)=P$ for large enough $n$, $l_n$ reaches its supremum for $\mu_1,\ldots,\mu_G\in K$ and is 
bounded away from the supremum if 
not all of $\mu_1,\ldots,\mu_G\in K$, a.s.
\end{description}
\end{lemma}
\begin{proof}
Start with part (a). Consider a sequence $\set{\theta_m}_{m\in\N}\in \Theta_G(P)^\N$ with  $\norm{\mu_{jm}}\to\infty$ for $j=1,\ldots,k,\ 1\le k<G$ and a compact set $K$ with $\mu_{jm}\in K$ for $j>k$. Let 
\begin{equation*}
A_m=\set{x: \ \forall j\in\set{1,\ldots,k}: \phi(x,\mu_{j,m},\Sigma_{j,m}) \le  \epsilon_m\sum_{l=k+1}^G\pi_{l,m}\phi(x,\mu_{l,m},\Sigma_{l,m})}, 
\end{equation*}
where  $\epsilon_m\searrow 0$ slowly enough that $P(A_m)\to 1$. Let $\pi^*_m=\sum_{j=1}^k\pi_{j,m}$. Let $\theta_{(G-k), m}$ for $m\in\N$ be defined by $\pi_{0,(G-k),m}= \pi_{0,m},\ \pi_{(j-k),(G-k),m}=\pi_{j,m}(1-\pi_{0,m})(1-\pi^*_m-\pi_{0,m})^{-1}$  for $j=k+1,\ldots,G$ accompanied by  the $\mu_j,\Sigma_j$-parameters belonging to the components  $k+1,\ldots,G$ of $\theta_m$.   Observe, using Remark \ref{rasphimax},
\begin{align*}
L(\theta_m)  %
   & =  \int_{A_m}\log(\psi_{\delta}(x,\theta_{m}))dP(x)+ \int_{A_m^c}\log(\psi_{\delta}(x,\theta_{m}))dP(x) \leq \\
   &\leq \int_{A_m}\log ((1+\epsilon_m)\psi_{\delta}(x,\theta_{(G-k)m})) dP(x) +P(A_m^c)\log (\delta+\phi_{max})
\end{align*}
implying $L(\theta_{(G-k)m})-L(\theta_m) \to 0$. $\E_P [\pi_{0,(G-k),m}\delta (\psi_{\delta}(x,\theta_{(G-k)m}))^{-1}] \le\pimax$ will be fulfilled for $m$ large enough because it is fulfilled for $\theta_m$  by definition and $\psi_{\delta}(x,\theta_{(G-k),m})>\psi_{\delta}(x,\theta_{m})$ on $A_m$ with $P(A_m)\to 1$. $L(\theta_{(G-k)m})<L_{G-1}$ implies that, because of A2, $\theta_m$ is bounded away from 
$L_G$. 

Regarding existence of a maximum with $\mu_1,\ldots,\mu_G\in K$, observe that with Remark \ref{rasphimax}, 
$\psi_\delta(x,\theta)$ can be bounded by $\delta+\phi_{max}$ for 
all $\theta$ for which $L(\theta)>L_G-\epsilon$. Now consider a sequence 
$(\theta_m)_{m\in\N}$ so that $\forall m:\ \mu_{1m},\ldots,\mu_{Gm}\in K$,
with the notation of Lemma \ref{laslambda},
$\lambda^*_{min}\le \lambda_{min}(\theta_m)\le \lambda_{max}(\theta_m)\le \lambda^*_{max}$ and $L(\theta_m)\to L_G$. Because of compactness, w.l.o.g., $\theta_m\to\theta_+$ and, using Fatou's Lemma, $L_G=\lim_{m\to\infty} L(\theta_m)
\le E_P\limsup_m\psi_\delta(x,\theta_m)=L(\theta_+)\le L_G$.

Part (b) holds because if $(\theta_m)_{m\in\N}$ is chosen as above for 
$m=n\to\infty$ and $P$ is replaced by the empirical distribution $P_n$, Glivenko-Cantelli enforces $P_n(A_{n})-P(A_{n})\to 0$ a.s. \CHANGE{ 
Glivenko-Cantelli applies here because a sequence of closed balls 
$(B_n)_{n\in \N}$ can be constructed so that $B_n\subseteq A_n,\ P(B_n)\to 1$ 
a.s.;
the closed balls are a Vapnik-Chervonenkis class, and $P_n(A_{n})\ge P_n(B_n)\to
1$ a.s.}
Furthermore, for $\theta\in \Theta_G(P)$ with $L(\theta)>L_{G-1}$: $l_n(\theta)\to L(\theta)$ a.s. by the strong law of large numbers, so that for large enough $n$: $\sup_{\theta\in\Theta_G(P)}l_n(\theta)>L_{G-1}$ a.s. On the other hand,  
$\theta_{(G-k),n}$ can be chosen optimally in a compact set $K$ because of Lemma
\ref{laslambda}, within which $l_n$ converges uniformly to $L$ a.s.
(Theorem 2 in \cite{Jennrich_1969}), and 
therefore,  
$\limsup_{n\to\infty} l_n(\theta_{(G-k),n})\le L_{G-1}$. With these ingredients, the argument of part (a) carries over.
\end{proof}
\begin{lemma}\label{lasmu}  
Assume A1 and A3. There is a compact set $K\subset\R^p$ so that 
\begin{description}
\item[(a)] $L$ reaches its supremum for $\mu_1,\ldots,\mu_G\in K$,
\item[(b)] for $x_1,x_2,\ldots$ i.i.d. with ${\cal L}(x_1)=P$, there exists a
sequence $(\tilde \theta_n)_{n\in\N}$ maximizing $l_n$ locally for $\mu_1,\ldots,\mu_G\in K$ so that $l_n(\tilde \theta_n)\to L_G$ a.s., and a.s. there is no sequence $\theta_n\in \Theta_G(P)$ so that 
$\limsup_{n\to \infty} l_n(\theta_n)> L_G$.
\end{description}
\end{lemma}
\begin{proof}
Start with part (a). Consider a sequence $\set{\theta_m}_{m\in\N} \in \Theta_G(P)^\N$ with $\norm{\mu_{j,m}} \to \infty$ for $j=1,\ldots,k,\ 1\le k<G$ (the case $k=G$ is treated at the end), and a compact set $K$ with $\mu_{j,m}\in K$ for $j>k$. By selecting a subsequence if necessary,  assume that there exists $\mu_j=\lim_{m\to\infty} \mu_{j,m},$ $\Sigma_j=\lim_{m\to\infty}\Sigma_{j,m}$ for $k+1\le j\le G$ and that $\pi_{j,m}$
converge for $j=0,\ldots,G$. Let $j^*=\argmax_{k+1\le j\le G}\E_P \phi(x,\mu_{j},\Sigma_{j})$. Suppose $L(\theta_m) \nearrow L_G$ monotonically and, by A3, assume $\pi_{0,m}>\epsilon_1$.

Consider first the case $\sum_{j=1}^k\pi_{j,m}\to\epsilon_3>0$. Construct another sequence $\set{\theta_{*m}}_{m\in\N} \in \Theta_G(P)^\N$ for which  $\mu_{1,*m}=\ldots=\mu_{k,*m}=\mu_{j^*}\in K,$ $\Sigma_{1,*m}=\ldots=\Sigma_{k,*m}=\Sigma_{j^*}$. All other parameters are the  same as in $\theta_m$. Let $A_m=\set{x:\ \forall j\in\{1,\ldots,k}: 2\phi(x,\mu_{j,m},\Sigma_{j,m})\le \phi(x,\mu_{j^*},\Sigma_{j^*})\}$.  Observe  $P(A_m)\to 1$. Now
\begin{equation}\label{easll}
L(\theta_{*m})-L(\theta_m)= %
\int_{A_m}\log\frac{\psi_{\delta}(x,\theta_{*m})} {\psi_{\delta}(x,\theta_{m})}dP(x) 
+ \int_{A_m^c}\log\frac{\psi_{\delta}(x,\theta_{*m})} {\psi_{\delta}(x,\theta_{m})}dP(x).
\end{equation}
For large enough $m$, 
\begin{equation*}
\int_{A_m} \log \frac{\psi_{\delta}(x,\theta_{*m})} {\psi_{\delta}(x,\theta_{m})}dP(x) \ge \epsilon_5>0,
\end{equation*}
whereas (using Remark \ref{rasphimax})
\begin{equation} \label{eq:nullc}
\int_{A_m^c}\log\frac{\psi_{\delta}(x,\theta_{*m})} {\psi_{\delta}(x,\theta_{m})}dP(x) \ge  P(A_m^c)\log\frac{\epsilon_1\delta}{\delta+\phi_{\max}} \to 0.
\end{equation}
Therefore $L(\theta_{*m})-L(\theta_m)>0$ for large enough $m$ so that $L(\theta_m)$ is improved by a $\theta$ with $\mu_j\in K$ for $\js$.\\
Consider now $\sum_{j=1}^k\pi_{jm}\to 0$.  Construct another sequence  $\set{\theta_{*m}}_{m\in\N}\in \Theta_G(P)^\N$ for which  $\mu_{1,*m}=\ldots=\mu_{k,*m}=\mu_{j^*}\in K,$ $\Sigma_{1,*m}=\ldots=\Sigma_{k,*m}=\Sigma_{j^*}$,  $\pi_{1*m}=\ldots=\pi_{k*m}=0$, $\pi_{0*m}=\sum_{j=0}^k\pi_{jm}$, all other parameters taken from $\theta_m$. Set $A_m=\{x:\ \forall j\in\{1,\ldots,k\}: \phi(x,\mu_{jm},\Sigma_{jm})<\delta\}$. Again $P(A_m)\to 1$. With this, 
\eqref{easll} holds again. This time
\begin{equation*}
\int_{A_m}\log\frac{\psi_{\delta}(x,\theta_{*m})} {\psi_{\delta}(x,\theta_{m})}dP(x)>0
\end{equation*}
and again \eqref{eq:nullc}.

Let $\theta_*=\lim_{m\to\infty}\theta_{*m}$ (this exists by construction). Continuity of $L$ implies that $L(\theta_*)=L_G$ and therefore for all $m:$  $L(\theta_*)\ge L(\theta_m)$.  $\sum_{j=1}^k\pi_{j,m}\to 0$  is required  here because $\pi_{0,*m}\ge \pi_{0,m}$ does not necessarily fulfill $\E_P\frac{\pi_{0,*m}c}{\psi_{\delta}(x,\theta_{*m})}\le \pimax$, but  $\lim_{m\to\infty}\pi_{0,*m}=\lim_{m\to\infty}\pi_{0,m}$ does.  \\
Finally, consider $k=G$. With $A_{m,\epsilon}=\{x:\ \forall j\in\{1,\ldots,k\}: \phi(x,\mu_{j,m},\Sigma_{j,m})<\delta \epsilon\}$, observe 
\begin{equation*}
\E_P\frac{\pi_{0,m}\delta}{\psi_{\delta}(x,\theta_{m})}\ge P(A_{m,\epsilon})  \frac{\pi_{0,m}\delta}{\pi_{0,m}(c+\epsilon)}>\pimax,
\end{equation*}
for small enough $\epsilon$ and large enough $m$,  violating for large $m$ the corresponding constraint in $\Theta_G(P)$ as long as $\pi_{0,m}$ is bounded from below, as was assumed. Existence follows in the same way as in the proof of Lemma \ref{lasmua3}.

For part (b) let $\theta^*$ have $\mu_1^*,\ldots,\mu_G^*\in K$ and $L(\theta^*)=L_G$, which exists because of part (a) and Lemma \ref{laslambda}, which ensures further that $\theta^*$ is in a compact $K^*$. Then the strong law of large numbers yields $l_n(\theta^*)\to L_G$ a.s., and Theorem 2 of \cite{Jennrich_1969} implies that for all sequences $(\theta_n)_{n\in\N}\in (K^*)^\N:$ $\limsup_{n\to \infty} l_n(\theta_n)\le L_G$. This also holds for sequences $(\theta_n)_{n\in\N}$ that are eventually outside $K^*$ because of part (a) of Lemma \ref{laslambda} and the proof of part (a) above, because if $(\theta_m)_{m\in\N}$ is chosen as above 
for $m=n\to\infty$ and $P$ is replaced by the empirical distribution $P_n$, 
Glivenko-Cantelli \CHANGE{(which applies by the same argument as used
in the proof of Lemma \ref{lasmua3})} enforces $P_n(A_{n})-P(A_{n})\to 0$ a.s., which means that as in part (a), a.s., eventually $l_n(\theta_n)$ cannot converge to anything larger than $L_G$.
\end{proof}

\begin{theorem}[RIMLE existence]\label{th:asy_existence}
Assume A1 and any one of A2 or A3. There is a compact subset  $K\subset \Theta_G(P)$ so that there exists $\theta\in K:\ \infty>L(\theta)=L_G>-\infty$. Assuming A2, for $\theta\not\in K$, $L(\theta)$ is bounded away from $L_G$.
\end{theorem}
\begin{proof} 
Pieced together from Lemmas  \ref{lasexist}-\ref{lasmua3} parts (a) and 
Remark \ref{rasphimax}.
\end{proof}
Theorem \ref{th:asy_existence} establishes existence of the RIMLE functional
\begin{equation}\label{eq_rimle_functional}
\theta^\star(\delta) = \argmax_{\theta \in \Theta_G(P)} L(\theta).
\end{equation}
Unfortunately neither $L(\theta)$ nor $l_n(\theta)$ can be expected to have a unique maximum. If we take the vector $\theta$ and we permute some of the triples $(\pi_j, \mu_j, \Sigma_j)$ we still obtain the same value for $L(\theta)$ and $l_{n}(\theta)$. This known as ``label switching'' in the mixture literature. There could be other causes for multiple maxima. Without strong restrictions on $P$, we cannot identify any specific source of multiple optima in the target function. Instead we show that asymptotically the sequence of estimators is close to some maximum of the pseudo-loglikelihood, which amounts to  consistency of the RIMLE with respect to a quotient space topology identifying all loglikelihood maxima, as done in \cite{Redner_1981}. 
By $\theta^\star(\delta)$ in \eqref{eq_rimle_functional} we mean any of the maximizer of $L(\theta)$. Define the sets
\begin{equation*}
S(\dot \theta)=\set{ \theta \in \Theta_G(P) : \int \log \psi_\delta (x;\theta)dP(x) = \int \log \psi_\delta (x;\dot \theta)dP(x) },
\end{equation*}
\begin{equation*}
\mathcal{K}(\dot \theta, \eps)= \set{\theta \in \Theta_G(P) : \|\theta - \ddot{\theta} \| <  \eps \;  \forall \; \ddot \theta \in S(\dot \theta)},
\quad \text{for any}  \quad \eps >0.
\end{equation*}
The following theorem makes a stronger statement assuming A2 than A3, because 
if A2 does not hold, the $G$th mixture component is asymptotically not needed 
and cannot be controlled for finite $n$ outside a compact set.
\begin{theorem}[Consistency]\label{th:rimle_consistency}
Assume A1 and A2. Then  for every $\eps >0 $ and every sequence of maximizers
$\theta_n (\delta)$ of $l_n$:  
\begin{displaymath}
\lim_{n\to\infty}P\set{ \theta_n (\delta)  \in \mathcal{K}(\theta^\star(\delta), \eps) } =1.
\end{displaymath}
Assuming A3 instead of A2, for every compact $K\supset \mathcal{K}(\theta^\star(\delta), \eps)$ there exists a sequence of $\theta_n$ that maximize $l_n$ 
locally in $K$ so that
\begin{displaymath}
\lim_{n\to\infty}P\set{ \theta_n   \in \mathcal{K}(\theta^\star(\delta), \eps) } =1.
\end{displaymath}
\end{theorem}
\begin{proof}
Under A2, because of the parts (b) of the Lemmas \ref{laslambda} and \ref{lasmua3} it can be assumed that there is a compact set $K$ so that all $\theta_n(\delta)\in K$ for large enough $n$ a.s. Under A3, considerations are restricted to $K$ anyway.

Based on Theorem  \ref{th:asy_existence} and related Lemmas $|\log \psi_\delta(x, \theta) | \leq C$ for some finite constant $C$  for all $\theta \in K$.  Sufficient conditions for Theorem 2 in \cite{Jennrich_1969} are satisfied, which implies uniform convergence of $l_n(\theta)$, that is  $\sup_{\theta \in K} \abs{l_{n}(\theta)-L(\theta)} \conv 0$ $P$-a.s. Based on the latter, and  applying the  same argument as in proof of Theorem 5.7 in \cite{Vaart_2000}, it holds true that  $L({\theta}_n(\delta))  \conv  L({\theta}^\star(\delta))$ $P$-a.s. By continuity of $L(\theta)$ and Theorem \ref{th:asy_existence}  we have that for every $\eps>0$ there exists a $\beta>0$ such that $L(\theta^\star(\delta))-L(\theta)>\beta$ for all  $\theta \in K\setminus\mathcal{K}(\theta^\star(\delta), \eps)$.  Denote $(\Omega, \mathcal{A}, P)$ the probability space where the sample random variables are  defined and consider the following events
\begin{displaymath}
A_n=\set{\omega \in \Omega : \; \theta_n(\delta)  \in K\setminus\mathcal{K}(\theta^\star(\delta),\eps)}, 
\end{displaymath} 
and 
\begin{displaymath}
B_n = \set{\omega \in \Omega : \; L(\theta^\star(\delta)) -  L(\theta_n(\delta))> \beta}.
\end{displaymath} 
Clearly $A_n \subseteq B_n$ for all $n$. $P(B_n) \to 0$ for $n \to \infty$ implies  $P(A_n) \to 0$. The latter proves the result.
\end{proof}


\section{Algorithms and practical issues}\label{sec_practical}

\subsection{RIMLE computing}\label{subsec_rimle_computing}
In this section we develop  Expectation--Maximization type algorithms (EM) to compute the RIMLE (for a fixed $\delta$ ).  Let $s=0,1,\ldots$ be the iteration index. Let $a^\snew$ be the quantity $a$ computed at the $s$th step of the EM algorithm. Define 
\begin{align}\label{eq:Q_em}
Q(\theta, \theta^{(s)}) =& \sum_{i=1}^{n} \sum_{j=0}^{G} \tau_{j}(x_i, \theta^{(s)}) \log \pi_j + \; \sum_{i=1}^n \tau_0(x_i, \theta^{(s)})\log\delta + \nonumber \\ 
{}                     +& \sum_{i=1}^{n} \sum_{j=1}^{G}  \tau_{j}(x_i, \theta^{(s)}) \log \phi(x_i; \mu_j,\Sigma_j).%
\end{align}
Increasing \eqref{eq:Q_em} by an appropriate choice of $\theta$ increases $l_n(\cdot)$. An approximate candidate maximum of $l_{n}(\theta)$ can be found by the following EM--algorithm

\begin{algorithm}[H]
\caption{EM-algorithm}\label{algo:emfull}
\SetAlgoLined
\Input{$\{x_1,x_2,\ldots,x_n\}$, $\delta$, $\pimax$, $\gamma, \theta^{(0)}$, {\tt tol}>0}
\Output{$\theta^\emconv$ }
\BlankLine
\While{$|l_n(\theta^\snew) - l_n(\theta^\sold)| > ${\tt tol}}{
{\textsf{\textbf{E--step:}}} compute $\tau_j(x_i, \theta^\sold)$, for all $\is$ and $j=0,1,\ldots,G$\\
{\textsf{\textbf{M--step:}}} $\theta^\snew \gets {\argmax}_{\theta \in \Theta} Q(\theta, \theta^\sold)$ \\
}
$\theta^\emconv  \gets \theta^{(s+1)}$
\end{algorithm}

\begin{proposition}\label{prop-em-convergence-improper}
Assume A0. The sequence $\{\theta^{(s)}\}_{s \in \N}$ produced by Algorithm \ref{algo:emfull} converges to a point $\theta_n^{em} \in \Theta$, and $l_{n}(\theta^{(s)})$ is increased in every step.
\end{proposition}
\begin{proof}
Find a set $A_n \subset \R^p$ that contains all points in   $\underline{x_n}$ with  Lebesgue measure
$M(A_n) = 1/\delta$. $\delta$ is then  a proper uniform density function on $A_n$. Hence, for a given dataset the pseudo-density $\psi_\delta(\cdot)$ can be written as proper density function. Therefore the convergence Theorem 4.1  in \cite{Redner_Walker_1984} holds,  with $Q(\theta, \theta^{(s)})$ playing the role of their $Q(\cdot)$ function.
\end{proof}

Algorithm \ref{algo:emfull} is the analog of the EM algorithm for plain Gaussian mixtures  \citep[see][]{Redner_Walker_1984} except that now the M-step is a constrained optimization. \cite{Wu_1983} showed that the  EM algorithm converges to the global maximum  if the likelihood function is unimodal and certain differentiability conditions are satisfied. In general the limit of the  EM algorithm   is not guaranteed to coincide with a global maximum of likelihood function. However, Proposition \ref{prop-em-convergence-improper} guarantees that $\theta_n^\emconv$ is a stationary point of $l_{n}(\cdot)$. Running the EM algorithm for a large number of starting values increases the chances of finding the optimal solution. For finite  Gaussian mixtures models it is well known that the likelihood surface is difficult to explore even when $p$ is not too large, and the main advantage of the EM algorithm is that the M-step can be divided in a number of simpler optimization problems each of which has a closed form solution. However, for the RIMLE  the constraints add some complexity, and in particular the noise proportion constraint does not allow to separate the M-step in simpler subprograms. One possibility is to perform the M-step using  numerical optimization packages, but the eigenratio constraints requires to  parameterize each $\Sigma_j$  terms of its spectral components. The latter has the drawback to add $G{\times}p(1-p)/2$ parameters. Furthermore, the eigenratio constraint has a non-smooth nature that would make numerical techniques hard to adapt.

In \cite{Coretto_Hennig_2014_comparison} computations are based on Algorithm \ref{algo:emfull} where the M-step is performed as if the problem would be unconstrained, and  breaking the iteration  when updates drive the parameters outside the constrained parameter space. \cite{Coretto_Hennig_2014_comparison} also propose a heuristic method to enforce the constraints at the end of the iterations if necessary. Of course in such situations there would be no guarantee that the delivered solution is a stationary point of $l_n(\cdot)$. Here, we propose an algorithm where constraints are applied exactly in each iteration. The M-step in Algorithm \ref{algo:emfull} is replaced with two conditional maximization (CM) steps. This transforms Algorithm \ref{algo:emfull} into an  Expectation-Conditional Maximization algorithm (ECM) as introduced  by  \cite{Meng_Rubin_1993}. For ease of notation, for  $j=0,1,\ldots,G$ define
$\toldj = \tau_{j}(x_i, \theta^\sold)$ and $\Toldj = \sumin \tau_{j}(x_i, \theta^\sold)$.
Rewrite \eqref{eq:Q_em}, using $\theta_1=(\mu_1,\ldots,\mu_G,\vect(\Sigma_1),\allowbreak \ldots,\allowbreak \vect(\Sigma_G))$, $\theta_2=(\pi_0,\pi_1,\ldots,\pi_G)^\prime$, and $\theta=(\theta_1,\theta_2)^\prime$, as 
\begin{equation}\label{eq:Q12_em}
Q(\theta_1,\theta_2, \theta^\sold)= Q_1(\theta_1, \theta^\sold) %
+ Q_2(\theta_2, \theta^\sold) + \text{const},
\end{equation}
where
\begin{equation*}
Q_1(\theta_1, \theta^\sold)= \sumin \sumjG \toldj \log \phi(x_i; \mu_j,\Sigma_j),\quad 
Q_2(\theta_2, \theta^\sold)= \sum_{j=0}^G \Toldj  \log \pi_j, 
\end{equation*}
and  $\text{const} =\Toldo \log(\delta)$ which does not depend on $\theta$. 
Consider the following programs: 
\begin{equation}
\underset{\theta_1}{\text{maximize}} \;  Q_1(\theta_1, \theta^\sold)
\quad \text{subject to} \quad \frac{\lmax(\theta_1)}{\lmin(\theta_1)} \leq \gamma,  \nonumber  \label{eq:cm1}\tag{\bf CM1}
\end{equation}
and
\begin{equation}
\underset{\theta_2}{\text{maximize}} \;  Q_2(\theta_2, \theta^\sold) %
\quad \text{subject to} \quad 
\sum_{j=0}^G \pi_j=1, 
\;  \sumin \tau_0(x_i,\theta_1^\snew, \theta_2) \leq n\pimax.
\label{eq:cm2}\tag{\bf CM2}
\end{equation}
The ECM algorithm consists of solving \eqref{eq:cm1} and then \eqref{eq:cm2}. The sequence of optimizations replaces the  M-step in Algorithm \ref{algo:emfull}.  Notice that in \eqref{eq:cm2}, $\pi_j=0$ for some $j=0,1,\ldots,G$ would drive the objective function toward $-\infty$, so we do not need to restrict the $\pi_j$ as $>0$. $\Toldj=0$ will not happen, see Remark \ref{remark:ecm_eat_component}. Also notice that for $\delta{=}0$ the noise proportion constraint is automatically fulfilled, and for more analysis on these cases see Remark \ref{remark:ecm_eat_component}.   
\begin{algorithm}[!t]
\caption{ECM}\label{algo:ecm}
\SetAlgoLined
\Input{$\{x_1,x_2,\ldots,x_n\}$, $\delta$, $\pimax$, $\gamma, \theta^{(0)} \in \Theta$, {\tt tol}>0}
\Output{$\theta^\text{ecm}$ }
\BlankLine
\While{$|l_n(\theta^\snew) - l_n(\theta^\sold)| > ${\tt tol}}{
\BlankLine
\BlankLine
{\textsf{\textbf{E--step}}} \newline 
compute $\toldj$ for all $\is$ and $j=0,1,\ldots,G$
\BlankLine
\BlankLine
{\textsf{\textbf{CM1--step}}} \newline 
\For{$\js$}{
$\mu^\snew_j \gets  \frac{1}{\Toldj} \sum_{i=1}^{n} \toldj  x_i $\newline
$S^\snew_j   \gets  \frac{1}{\Toldj} \sum_{i=1}^{n} \toldj  \;  (x_i - \mu_{j}^\snew) (x_i-\mu_{j}^\snew)', $\newline
$ V^\snew_j  \diag(e_{j,1},\ldots,e_{j,p})   V^{\snew\prime}_j \quad \gets \quad \spec(S^\snew_j)$
}
\uIf{$\max\left\{{e_{j,k}}/{e_{t,k}};  \quad  t,\js, \; j \neq t, \; \ks \right\}\leq \gamma$}{$\Sigma_j^\snew \gets S^{(s+1)}_j$}
\Else{
\mathleft 
\begin{equation*} 
m_* \gets  \argmin_{m >0} \quad \sum_{j=1}^G T_j^{(s)} \sum_{k=1}^p 
\left( \log(\ell_\gamma(e_{j,k},m)) + \frac{e_{j,k}}{\ell_\gamma(e_{j,k},m)}\right)
\end{equation*} 
\mathcenter
$\Sigma_j^\snew \gets V^\snew_j\diag(\ell_\gamma(e_{j,1},m_*), \ldots, \ell_\gamma(e_{j,p},m_*)) \; V^{\snew\prime}_j$ for all $\js$ 
}
\BlankLine
\BlankLine
{\textsf{\textbf{CM2--step}}} \newline 
\uIf{$\sumin \tau_0(x_i, (\theta_1^\snew, \dot \theta_2))\leq n\pimax$,
~~\textrm{where}  $\dot \theta_2=\left(\Toldo/n,\ldots, T^\sold_G/n \right)^\prime$}%
{$\pi_j^\snew \gets \Toldj/n$ for all $j=0,1,\ldots,G$}
\Else{
\mathleft 
\begin{equation*}
\text{compute~}\omega_*:\quad %
\left( \sumin \frac{\omega_* \delta}{\omega_* \delta + \frac{1-\omega_*}{n-\Toldo} \sumjG \Toldj \phi(x_i; \mu_j^\snew, \Sigma_j^\snew)}\right) - n\pimax = 0
\end{equation*}
\mathcenter
$\pi_0^\snew \gets \quad \omega_*$\newline 
$\pi_j^\snew \gets  \quad  \frac{1-\omega_*}{n-\Toldo} \Toldj $\newline

}

}$\theta^{\text{ecm}}  \gets \theta^{(s+1)}$
\end{algorithm}


Before presenting the ECM Algorithm \ref{algo:ecm} we introduce additional notations. For $a \in R^d$ let $\diag(a)$ be the $d{\times}d$ diagonal matrix with elements of $a$ on the main diagonal. For a matrix $A$ let $\spec(A) = \Gamma \Lambda \Gamma^\prime$ be the spectral decomposition of $A$, that is, $\Gamma$ contains the normalized eigenvectors of $A$ corresponding to the eigenvalues contained in the diagonal matrix $\Lambda$. Moreover for $a,m \in \R$ define the shrinkage operator $\ell_{\gamma}(a,m) = \min\{ \max \{m, a\}, \gamma m \}$.   \\
In each step of Algorithm \ref{algo:ecm} closed form expressions are computed except that for computing $m_*$ in (CM1), and $\omega_*$ in (CM2).  $m_*$ is the solution of a one-dimensional convex problem. The resulting updates for the eigenvalues  almost coincide with those of TCLUST. For TCLUST \cite{Fritz_etal_2013} show  that their analog of  $m_*$  can be computed by $2pG+1$ evaluations of the objective function. A similar result may hold here, however we do not consider it because  we found, \CHANGE{based on experimental evidence},  that the simple golden section search algorithm of \cite{Kiefer_1953} finds $m_*$ in less than 30 \CHANGE{objective function} evaluations on average independently of $p$ and $G$. Computation of  $\omega_*$  can be performed by a one-dimensional  root finder algorithm. Both are simple problems that do not require much computational effort.

Some additional results are given to show how the {CM1--step} and  the {CM2--step} solve \eqref{eq:cm1} and \eqref{eq:cm2} respectively.
\begin{lemma}\label{lemma:cm1_step}
Assume Algorithm \ref{algo:ecm}  has been run for $s$ iterations. The vector
$\theta_1^\snew=(\mu_1^\snew,\ldots,\mu_G^\snew, \vect(\Sigma_1^\snew),\ldots, \vect(\Sigma_G^\snew))^\prime$ computed in the CM1--step is the global optimal solution to  \eqref{eq:cm1}. Moreover,  $m_*$ exists and it is unique.
\end{lemma}
\begin{proof}
Based on standard normal likelihood theory, one can see that the unique maximum of 
$Q_1(\theta_1, \theta^\sold)$  with respect to mean parameters is $\mu_j^{(s+1)}$ for all $\js$. Substituting $\mu_j^{(s+1)}$ into $Q_1(\cdot)$, and rearranging the exponent of the Gaussian density by using the cyclic property of the matrix trace \citep[see][]{Anderson_Olkin_1985}, program \eqref{eq:cm1} is completed by choosing $\Sigma_1,\ldots,\Sigma_G$ maximizing 
\begin{equation}\label{eq_em_Q1_2}
Q_1(\dot \theta_1, \theta^\sold) = \text{const} - \frac{1}{2} \sumjG  \Toldj  \left( \tr(\Sigma_j^{-1} S_{j}^\snew) - \log \det(\Sigma_j^{-1}) \right)
\end{equation}
under the eigenratio constraint, where $\dot \theta_1 = (\mu_1^\snew, \ldots, \mu_j^\snew, \vect(\Sigma_1), \ldots, \vect(\Sigma_G))$. Consider the spectral decompositions 
\begin{equation*}
\spec(\Sigma_j)=\Gamma_j \Lambda_j \Gamma_j^\prime, %
\quad \text{and} \quad %
\spec(S^\snew_j)=V^\snew_j  E_j   V^{\snew\prime}_j,
\end{equation*}
where $\Lambda_j=\diag(\lambda_{j,1},\ldots,\lambda_{j,p})$ and  $E_j=\diag(e_{j,1},\ldots,e_{j,p})$. Theorem 1 and Corollary 1 in \cite{Theobald_1975} imply that
\begin{equation}
\tr\left(\Lambda_j^{-1} E_j\right) - \log\det(\Lambda_j^{-1})  \leq  \tr\left( \Sigma_j^{-1} S^{(s+1)}_j \right) - \log \det(\Sigma_j^{-1}),
\end{equation}
with the previous  holding with equality if and only if $\Gamma_j=V^{(s+1)}_j$.  Therefore, $\Sigma_j^\snew = V_j^\snew \Lambda_j V_j^{\snew\prime}$ is plugged into \eqref{eq_em_Q1_2}, and  \eqref{eq:cm1} reduces to 
\begin{equation}\label{eq_mstep_eigenvalues}
\begin{aligned}
& \underset{\lambda_{1,1}, \ldots, \lambda_{G,p} }{\text{minimize}}    & & \sum_{j=1}^G T_j^{(s)} \sum_{k=1}^p  \left( \log(\lambda_{j,k}) + \frac{e_{j,k}}{\lambda_{j,k}} \right), \\
& \text{subject to}      & & 0 < m  \leq \lambda_{j,1},\ldots, \leq \lambda_{j,p} \leq m\gamma \quad \forall \js.
\end{aligned}
\end{equation}
Program \eqref{eq_mstep_eigenvalues} is separable in the optimization variables, and therefore the summands of \eqref{eq_mstep_eigenvalues} can be minimized separately for a given $m$. Fix $m>0$, then $\ell_\gamma(e_{j,k},m)$ is the unique optimal solution to the minimization of  $\log(\lambda_{j,k}) + e_{j,k} \; \lambda_{j,k}^{-1}$. Notice that
$e_{j,k} \leq e_{j,t}  {\Longrightarrow} \ell_\gamma(e_{j,k},m) \leq \ell_\gamma(e_{j,t},m)$ for any $m>0$ and $t=1,2\ldots,p$. This means that the relative ordering of the elements on the diagonal of $E_j$ remains unchanged after having applied  the shrinkage operator $\ell(\cdot)$.  Replace $\lambda_{j,k}$ with $\ell_\gamma(e_{j,k},m)$ and \eqref{eq_mstep_eigenvalues} is transformed into 
\begin{equation}\label{eq_cem2_m}
\begin{aligned}
& \underset{m}{\text{minimize}}    & &  
\sum_{j=1}^G T_j^{(s)} \sum_{k=1}^p \left( \log(\ell_\gamma(e_{j,k},m)) + \frac{e_{j,k}}{\ell_\gamma(e_{j,k},m)}\right),\\
& \text{subject to}      & & m>0
\end{aligned}
\end{equation}
\eqref{eq_cem2_m} is now a convex  program in $m$. Therefore \eqref{eq:cm1} is solved by the unique $m_*$ that solves \eqref{eq_cem2_m}. This implies that the CM1--step is solved by taking   $\Sigma_j^\snew = V^\snew_j E_j^* V^{\snew\prime}_j$ where
$E_j^* = \diag(\ell_\gamma(e_{j,1},m_*), \ldots, \ell_\gamma(e_{j,p},m_*)).$ Notice that uniqueness of $m_*$ implies the uniqueness of the solution to \ref{eq:cm1}. Observe that when the eigenvalues of $S_1^\snew, \allowbreak \ldots, S_G^\snew$ fulfill the eigenratio constraint, then  $E_j^* = E_{j}$ and $\Sigma_j^\snew = S_j^\snew$ for all $j=1,\ldots,G$. The latter  completes the proof.   (The result is connected to Lemma 1 in \cite{Won_Lim_Kim_Rajaratnam_2013}, by which the last part of the proof is inspired.)  

\end{proof}

Based on the previous Lemma, the constrained eigenvalues can be found by simply solving a convex one-dimensional problem. The optimal choice of the covariances  is a form of Steinian-type nonlinear shrinkage \citep[see][]{Gavish_Donoho_2014}.

\begin{lemma}\label{lemma:cm2_step}
Assume Algorithm \ref{algo:ecm}  has been run for $s$ iterations.  The vector
$\theta_2^\snew=(\pi_0^\snew,\ldots,\pi_G^\snew)^\prime$ computed in the CM2--step is the global optimal solution to  \eqref{eq:cm2}. Moreover,  $\omega_*$ exists and it is unique.
\end{lemma}
\begin{proof}
The objective function in \eqref{eq:cm2} is strictly concave and  the equality constraint is linear. Take $\theta_2^\prime, \theta_2^{\prime\prime} \in \{\theta_2:\; \pi_0+\ldots,+\pi_G=1, \; \pi_j \in [0,1] \; \forall j=0,1,\ldots, G\}$, if $\pi_0^\prime \leq \pi_0^{\prime\prime}$
then $\tau_0(\cdot, \cdot, \theta_2^{\prime}) \leq \tau_0(\cdot, \cdot, \theta_2^{\prime\prime})$. Therefore, for $\beta \in (0,1)$ 
\begin{equation*}
\sumin \tau_0(x_i, \theta_1^\snew,\; \beta\theta_2^\prime+(1-\beta)\theta_2^{\prime\prime}) %
\leq \max\left\{ 
\sumin \tau_0(x_i, \theta_1^\snew, \theta_2^\prime), 
\sumin \tau_0(x_i, \theta_1^\snew, \theta_2^{\prime\prime})
\right\},
\end{equation*}
which implies that $\sumin \tau_0(x_i, \theta_1^\snew, \theta_2)$ is quasiconvex in the optimization variable $\theta_2$. It is concluded that  Karush–Kuhn–Tucker (KKT) conditions are necessary for a global optimal solution \citep[see][]{Bertsekas_1999}. Such a  solution will be a stationary point of the Langrangean function
\begin{equation*}
H(\theta_2,h_1,h_2):= Q_2(\theta_2, \theta^\sold)  %
+ h_1 \left(1-\sum_{j=0}^G \pi_j \right) %
+ h_2 \left(n\pimax - \sumin \tau_0(x_i, \theta_1^\snew,\theta_2) \right),
\end{equation*}
where $h_1$ and $h_2$ are the dual variables. Let  $\nabla_j$ denote derivatives  with respect to the $j$ component of  $\theta_2$. Let $\theta_2^\star$ the optimal solution, then based on  KKT conditions there exists  $h_1^\star,h_2^\star$ such that the following  hold
\begin{equation} \label{eq:kkt_gradient}
\nabla_j Q_2(\theta_2^\star,\theta^\sold) - h_1^\star - h_2^\star \nabla_j\sumin \tau_0(x_i, \theta_1^\snew,\theta_2^\star) =0   \qquad \text{for all}\quad {j=0,1,\ldots,G},
\end{equation}
\begin{equation} \label{eq:kkt_complementary}
h_2^\star\left( \sumin \tau_0(x_i, \theta_1^\snew,\theta_2^\star) - n\pimax\right)=0, %
\quad h_2^\star \geq 0.
\end{equation}
First consider the case when the noise proportion constraint does not bind, that is $h_2^\star=0$. Than \eqref{eq:kkt_gradient} becomes $\Toldj/\pi_j^\star - h_1^\star = 0$ for all ${j=0,1,\ldots,G}$. Solving the latter for $\pi_j^\star$, using   the equality constraints and that  $\sum_{j=0}^G \Toldj = n$, it results that $h_1^\star = n $ and $\pi_j^\star = \Toldj /n$ for all ${j=0,1,\ldots,G}$. \\
Now assume that the noise proportion constraints binds, hence $h_2^\star>0$. Let $\pi_0^\star=\omega$ and rewrite the equality constraint as $\pi_1^\star+, \allowbreak \ldots, \allowbreak +\pi_G^\star = 1- \omega$.  Stationary points of $H(\cdot)$ satisfy  $\Toldj/\pi_j^\star - h_1^\star = 0$.  Solving the latter for $\pi_j^\star$ and  using  the equality constraints it results that $h_1^\star = \sumjG \Toldj / (1-\omega) $. Since  $\sum_{j=1}^G \Toldj = n-\Toldo$, then
\begin{equation*}
\pi_j^\star = \frac{1-\omega}{n-\Toldo}\Toldj \quad \text{for all} \quad \js.
\end{equation*}
Now the solution for $\js$ is a function of $\omega$, which can be determined by using the fact that the inequality constraints binds. Define  
\begin{equation*}
g(\omega) = \left( \sumin \frac{\omega \delta}{\omega  \delta + \frac{1-\omega }{n-\Toldo} \sumjG \Toldj \phi(x_i; \mu_j^\snew, \Sigma_j^\snew)}\right) - n\pimax.
\end{equation*}
$g(\omega)$ is bracketed on the interval $[0,1]$, in fact $g(0)=-n\pimax<0$ and $g(1)=n(1-\pimax)>0.$ Moreover $g(\omega)$ is continuous, and it can be easily verified that it's derivative is continuous and positive at any $\omega \in (0,1)$. This implies that there exists a unique $\omega_*$ such that $g(\omega_*) = 0 $. 
Setting $\pi_0^\star = \omega_*$ and replacing $\omega_*$ into $\pi_j^\star$ gives the optimal solution. We now compare the two solutions in terms of objective function, and we show that there is hierarchy between them. Define
\begin{equation*}
\dot \theta_2  =  \left(\frac{\Toldo}{n},\frac{T^\sold_1}{n}, \ldots,   \frac{T^\sold_1}{n}\right)^\prime, \qquad
\ddot \theta_2 = \left(\omega_*,\frac{1-\omega_*}{n-\Toldo}T^\sold_1 , \ldots,  \frac{1-\omega_*}{n-\Toldo}T^\sold_G\right)^\prime.
\end{equation*}
Using Wald's information inequality it can be shown that 
\[
\frac{\Toldo}{n}\log\left(\omega_*\right) + 
\sum_{j=1}^G \frac{\Toldj}{n} \log\left(\frac{1-\omega_*}{n-\Toldo}\Toldj \right) %
\;\leq  \; %
\sum_{j=0}^G \frac{\Toldj}{n} \log\left(\frac{\Toldj}{n}\right), %
\]
with the previous holding  with equality if and only if $\omega_*=\Toldo/n$. The latter implies that $Q_2(\ddot \theta_2, \theta^\sold) <  Q_2(\dot \theta_2, \theta^\sold)$ whenever $\ddot \theta_2 \neq \dot \theta_2$. Hence $\dot \theta_2$ is the global optimal solution whenever it is feasible, otherwise the global optimal solution is $\ddot \theta_2$. The latter proves that the  updating in CM2-step selects the  global optimal  solution to \eqref{eq:cm2}.
\end{proof}

\begin{theorem}\label{theorem:ecm}
Assume A0. The $\{\theta^\sold\}{s \in \N}$ produced by Algorithm \ref{algo:ecm} converges to a point $\theta_n^{\text{ecm}} \in \Theta$, and $l_{n}(\theta^{(s)})$ is increased in every step.
\end{theorem}
\begin{proof}
As consequence of Lemma \ref{lemma:cm1_step} and \ref{lemma:cm2_step},  $Q(\theta, \theta^\sold)$ is never decreased, in fact for all $s=0,1,\ldots$
\[
Q(\theta_1^\snew,\theta_2^\snew, \theta^\sold) \geq %
Q(\theta_1^\snew,\theta_2^\sold, \theta^\sold) \geq %
Q(\theta_1^\sold,\theta_2^\sold, \theta^\sold),
\]
A0 ensures existence of $\theta_n(\delta)$, and  the  convergence Theorem 4.1  in \cite{Redner_Walker_1984} holds  with $Q(\theta, \theta^{(s)})$ playing the role of their $Q(\cdot)$ function.
\end{proof}

\begin{remark}\label{remark:ecm_eat_component}
The eigenratio constraint together with the noise proportion constraint rule out the possibility that at some point along the iteration  $\Toldj=0$ for some $\js$ and updates in CM1-step are guaranteed to exists. In fact  $\Toldj=0$  means that according to  $\theta^{(s-1)}$ none of points contributes to the $j$th Gaussian component. In theory this  can only happens if the $j$th component has an infinite dispersion according to $\theta^{(s-1)}$. However, in that case  the eigenratio constraint would force all eigenvalues in  $\theta^{(s-1)}$  to diverge to $+\infty$ at the same rate so that   $\Toldj \searrow 0$  for all $\js$, which is not possible because of the noise proportion constraint. Although in theory an appropriate choice of $\theta^{(0)} \in \Theta$ should not produce such a degeneracy, it may well be that in practice this is caused because of limited numerical resolution.
Notice also that for $\delta{=}0$ the noise proportion constraint is automatically fulfilled, and this  would take the problem back to the EM algorithm for the MLE of a finite Gaussian mixture model with the additional eigenratio constraint.  Therefore Algorithm \ref{algo:ecm} would become the EM Algorithm \ref{algo:emfull} where the M-step would coincide with CM1-step of Algorithm \ref{algo:ecm} plus the usual updating for the proportion parameters: $\pi_j^\snew \gets \Toldj/n$ for all $j=0,1,\ldots,G$.
\end{remark}

\subsection{Choice of initial values and input parameters}\label{sec_choice_of_delta}

Algorithms \ref{algo:emfull} and \ref{algo:ecm} require the initial value $\theta^{(0)}$, and the input parameters $\pimax, \gamma$ and $\delta$. The initial value $\theta^{(0)}$ can be set by randomly assigning points to $G$ clusters and then computing cluster parameters. \CHANGE{ Initialization like this needs to be performed a number of times so that the solution with the largest pseudo-likelihood is selected. Implementation of the RIMLE given in the \pkg{otrimle} software of \cite{CRAN_otrimle} relies on a more refined initialization strategy which consist in the following steps.
\begin{description}
\item[Initial denoising:] for each data point compute its $k$th-nearest neighbors distance ($k$-NND), for some $k$. All points with $k$-NND larger than the $(1-\pimax)$-quantile of the 
$k$-NND are initialized as noise. The interpretation of $k$ is that  $(k-1)$, but not $k$, points close together may still be interpreted as noise/outliers, whereas $k$ such points would constitute a cluster. The default value in the \pkg{otrimle} package is $k=3$. 

\item[Initial clusters:] agglomerative hierarchical clustering based on ML criteria for Gaussian mixture models as in \cite{Fraley1998} is performed on the remaining $\lfloor n(1-\pimax)\rfloor$ regular points to find the initial clusters. The sample mean and covariance matrix of points belonging to each cluster are computed to define $\theta^{(0)}$. This step is performed based on the \texttt{hc()} function from the \pkg{mclust} package.
\end{description}
}
The constraint defining quantities $\pimax$ and $\gamma$ are regularization parameters that allow solving an otherwise ill-posed optimization problem. $\pimax$ also controls robustness because it specifies the maximum proportion of points assignable to the noise component. In order to be as robust as possible $\pimax=1/2$ is a convenient choice that guarantees maximum protection. This implements a familiar condition in robust statistics that at most half of the data should be classified as ``outliers/noise''. \CHANGE{ A choice of $\pimax$ lower than the actual noise/outlier proportion will enforce some outliers to be assigned to clusters with potentially problematic implications. Hence, unless one has prior knowledge about the contamination process, we suggested to stick to $\pimax=1/2$.}

The role of the eigenratio can be twofold. If $\gamma$ is set to a low-value, strong restrictions on clusters' shape are imposed. In this respect, the eigenratio constraint acts as a model selector. Unless one knows precisely the implications of a low choice of $\gamma$, it is suggested to use the eigenratio constraint as a regularization parameter. In fact, a large value of $\gamma$  will regularize the covariance matrices without affecting clusters' shape too much.  For example, a large $\gamma$ would allow discovering an elongated concentrated cluster along with clusters having widespread spherical scatters. \cite{Ritter_2014} contains an in-depth analysis of constraints in model-based clustering. \CHANGE{ In Section \ref{sec:experiment} we present Monte Carlo experiments where the effect of different $\gamma$ values is investigated.}

Although through the presence of the product $\pi_0\delta$ in (\ref{eq:psi}) the parameters $\pi_0$ and $\delta$ may seem confounded, they actually play a very different role in the RIMLE. $\delta$ is not treated as a model parameter to be estimated, but rather as a tuning device to enable a good robust clustering. The interpretation is that $\delta$ is the density value below which groups of observations should rather be treated as ``noise'' than as ``cluster''. This means that a larger value of $\delta$ will normally yield a larger estimate of $\pi_0$ because more observations will be classified as noise, as opposed to the intuition suggested by having the product $\pi_0\delta$ in (\ref{eq:psi}). Whether small groups of observations of a certain size and with a certain density peak should rather count as ``cluster'' or rather as ``group of outliers'' cannot be identified from the data alone, but is rather a matter of interpretation.  RIMLE may be sensitive to the choice of $\delta$, and a good choice of $\delta$ is therefore important in practice. \CHANGE{For instance, in the example of Figure \ref{fig:rimle_profiling_asynoise} it has been shown that outside a certain interval of $\delta$ values the RIMLE does not perform well.} Occasionally, subject matter knowledge may be available aiding the choice of a fixed value of $\delta$, but often such knowledge may not exist.  The OTRIMLE, a data dependent method (``optimally tuned RIMLE'') to choose $\delta$ is presented in \cite{Coretto_Hennig_2014_comparison}. The basic idea is to find a $\delta$ that optimizes a weighted Kolmogorov-type distance measure between the Mahalanobis distances of all objects to their corresponding cluster centers and the $\chi^2$-distribution, which the Mahalanobis distances should follow if the clusters were indeed Gaussian. \CHANGE{The current implementation of the OTRIMLE in the \pkg{otrimle} package selects the best RIMLE solution computed with algorithm \ref{algo:ecm} on a selected grid of 50 $\log(\delta)$ values. The default grid includes  $\log(\delta)=-\infty$ so that a pure Gaussian mixture is always included in the competition (see the \pkg{otrimle} manual for more details).}


\section{Breakdown robustness of the RIMLE}\label{sec_breakdown}

Although robustness results for some clustering methods can be found in the literature, robustness theory in cluster analysis remains a tricky issue. 
Some work exists on breakdown points  \citep{GarciaEscudero_Gordaliza_1999,Hennig_2004,Gallegos_Ritter_2005},  addressing whether parameters can diverge to infinity (or zero, for covariance eigenvalues and mixture proportions) under small modifications of the data. An addition breakdown point of $r/(n+r)$ means that $r$, but not $r-1$, points can be added to a data set of size $n$ so that at least one of the parameters ``breaks down'' in the above sense.

It is well known \citep{GarciaEscudero_Gordaliza_1999,Hennig_2008}, assuming the fitted number of  clusters to be fixed, that robustness in cluster analysis has to be data dependent, for the following reasons:
\begin{itemize}
\item If there are two not well separated clusters in the data set, a very  small amount of ``contamination'' can merge them, freeing up a cluster to fit outliers converging to infinity.
\item Very small clusters cannot be robust because a group of outlying points can  legitimately be seen as a ``cluster'' and will compete for fit  with non-outlying clusters of the same size. 
Noise component-based and trimming methods are prone to trimming whole clusters if they are small enough. 
\end{itemize}
Therefore, all nontrivial breakdown results (i.e., with breakdown point larger than the minimum $1/(n+1)$) in clustering require a condition that makes sure that the clusters in the data set are strongly clustered in some sense, which usually means that the clusters are homogeneous and strongly separated. 

The theory for the RIMLE given here 
generalizes the argument given in \cite{Hennig_2004}, Theorem 4.11, to the 
multivariate setup. 
We consider fixed datasets $\underline{x_n}=(x_1, x_2, \ldots, x_n)$ and
sequences of estimators $(E_{n})_{n\in\N}$ mapping 
observations from $(\R^p)^n$ to $\Theta$. Denote 
the components of $E_n(\underline{x_n})$ by\\ 
$(\pi_{En0},\pi_{En1},\ldots,\pi_{EnG}, \mu_{En1},\ldots,\mu_{EnG} ,\Sigma_{En1},\ldots,\Sigma_{EnG})$, $G$ being the number of mixture components as usually. 

The following assumption in the definition of the breakdown point
makes sure that $E_n(\underline{x_n})$ indeed parametrizes
$G$ different mixture components; if there was a mixture component with 
proportion zero or two equal ones, one mixture component would be free to 
be driven to breakdown. 
\begin{description}
\item[A4] For
$j=1,\ldots,G:$ $\pi_{Enj}>0$, and all $(\mu_{Enj},\Sigma_{Enj})$ are pairwise
different.
\end{description}
\begin{definition}\label{def:breakdown} 
Assume that $(E_n)_{n\in\N}$ and $\underline{x_n}$ fulfil A4. Then,
\begin{eqnarray*}
  B(E_n,\underline{x_n})& = & \min_r \left\{\frac{r}{n+r}:\ \exists 
1\le j \le G \right. \\
 && \forall D=[\pi_{min},1]\times C \mbox{ for which }\pi_{min}>0, 
C\subset \R^p \times {\cal S}_p \mbox{ compact} \\
 && \exists \underline{x_{n+r}}=(x_1, x_2, \ldots, x_n, x_{n+1},\ldots x_{n+r})
\mbox{ so that for }\\ 
&& \left.E_{n+g}(\underline{x_{n+g}}):\
(\pi_{E(n+g)j},\mu_{E(n+g)j}, \Sigma_{E(n+g)j})\not\in D\right\},
\end{eqnarray*}
where ${\cal S}_p$ is the set of all positive definite real valued
$p\times p$-matrices,
is called the {\bf breakdown point} of $E_n$ at dataset $\underline{x_n}$.
\end{definition}

Denote the sequence of RIMLE estimators
defined in \eqref{eq:rimle} as $(\theta_{mH})_{m\in\N}$, write
$l_{mH}(\underline{x_m},\theta)$ for $l_m(\theta)$ with any $m\in \N$ and
number of components $H$ in \eqref{eq:ln}, 
$l_{mH}^o=l_{mH}(\underline{x_m},\theta_{mH}(\underline{x_m}))$.
Let $\theta^*= \theta_{nG}(\underline{x_n})$ for the specific $\underline{x_n}$ and $G$ 
considered here.
Components of $\theta^*$ and later $\theta^+$ are denoted with upper index
``$*$'' and ``$+$'', respectively. For $j=1,\ldots,G$, let 
$\phi_j^*(x)=\phi(x,\mu_j^*,\Sigma_j^*)$, same with upper index ``$+$''. 
Assume $\delta>0$ fixed throughout this section.  
We start with a straightforward extension of Lemma \ref{lemma_finite_n_eigen_2}.
\begin{lemma} \label{lemma:unifxeigen}
Assume A0 for $\underline{x_n}$.  If $(\theta_m)_{m \in N}$ is any sequence in $\Theta$ so that  for some $\js$ and $\ks$,  $\lambda_{k,j,m} \searrow 0$  as $m \to \infty$. For $\underline{x_{n+r}}=(x_1, x_2, \ldots,\allowbreak x_n, x_{n+1},\ldots x_{n+r})$: $\sup_{(x_{n+1},\ldots,x_{n+r})\in(\R^p)^r} l_{(n+r)G}(\underline{x_{n+r}},\theta_m) \to -\infty.$
\end{lemma}
\begin{proof}
The proof of Lemma 
\ref{lemma_finite_n_eigen_2} still applies 
because adding further observations only adds further positive terms to the sum
in \eqref{eq_npc_small_eigenvalues}.
\end{proof}

\begin{corollary}\label{corollary:phimax} Assume that 
$\underline{x_n}$ is a fixed dataset fulfilling A0 and A4 for $E_n=\theta_{nG}$.
Then there is a $\lambda_0>0$ bounding from below all $\lambda_{min}$ for 
$\theta=\theta_{(n+r)G}(\underline{x_{n+r}})$ where 
$\underline{x_{n+r}}=(x_1, x_2, \ldots, x_n, x_{n+1},\ldots x_{n+r})$ for any $(x_{n+1},\ldots,x_{n+r})\in(\R^p)^r$. 
Consequently $\phi_{max}=(2\pi)^{-\frac{p}{2}}\lambda_0^{-\frac{p}{2}}$ is an upper 
bound for all $\phi(x;\mu,\Sigma)$ with $(\mu,\Sigma)$ occurring as component
parameters in any such $\theta$.
\end{corollary}
\begin{proof} 
\begin{equation} \label{eq:lnrbound}
\mbox{Observe }  l_{(n+r)G}(\underline{x_{n+r}},\theta^*)\ge \frac{1}{n+r} \left(
\sum_{i=1}^{n} \log \psi_\delta(x_i, \theta^*)+r\log(\pi_0^*\delta)\right)>-\infty.
\end{equation}
If the Corollary was wrong, it would be possible to construct a 
sequence $(\theta_m)_{m \in N}$ with $\lambda_{k,j,m} \searrow 0$ for some $\js$ and $\ks$ so that each $\theta_m=\theta_{(n+r)G}(\underline{x_{n+r}})$ for an admissible
$\underline{x_{n+r}}$. But \eqref{eq:lnrbound} implies that there is a lower
bound for $l_{(n+r)G}(\underline{x_{n+r}},\theta_m)$, contradicting Lemma 
\ref{lemma:unifxeigen}.
\end{proof}

The following theorem gives conditions under which the RIMLE estimator
is breakdown robust against adding $r$ observation to $\underline{x_n}$.
\eqref{eq:breakdowncon1} states that the dataset needs to be fitted by $G$
Gaussian components considerably better than by $G-1$ components, 
because otherwise the remaining mixture component would be available
for fitting the added observations without doing much damage to the original
fit. \eqref{eq:breakdowncon2} makes sure that the noise proportion in
$\underline{x_n}$ is low enough that the added observations can still be fitted
by the noise component without exceeding $\pi_{max}$.
\begin{theorem} \label{theorem:breakdown}
Assume that $\underline{x_n}$ fulfils A0 and A4 for 
$E_n=\theta_{nG}$. If
\begin{eqnarray}
  l_{n(G-1)}^o & < & \sum_{i=1}^n \log\left(\sum_{j=1}^G \pi_j^*
\phi_j^*(x_i)+\left(\pi_0^*+\frac{r}{n}\right)\delta\right)+
\nonumber\\
&& r\log\left(\left(\pi_0^*+\frac{r}{n}\right)\delta\right)+
(n+r)\log\frac{n}{n+r}-r\log \phi_{max}, \label{eq:breakdowncon1}
\end{eqnarray}
$\phi_{max}$ defined in Corollary \ref{corollary:phimax}, and 
\begin{equation}
  \label{eq:breakdowncon2}
\frac{1}{n+r}\left(\sum_{i=1}^n \frac{(n\pi_0^*+r)\delta}
{(n+r)\psi_\delta(x_i,\theta^*)} +r\right)<\pi_{max},
\end{equation}
then  $B(\theta_{nG},\underline{x_n})>\frac{r}{n+r}.$
\end{theorem}
\begin{proof}
For $\underline{x_{n+r}}=(x_1,\ldots x_{n+r})$, let 
$\theta^+=\theta_{(n+r)G}(\underline{x_{n+r}})$. Let $H<G$. Then,
\begin{displaymath}
  l_{(n+r)G}^o\le\sum_{i=1}^n\log\left(\sum_{j=1}^{H}
  \pi_j^+\phi_j^+(x_i)+\sum_{j=H+1}^G \pi_j^+\phi_j^+(x_i)+\pi_0^+\delta\right)+r\log \phi_{max}.
\end{displaymath}
Assume w.l.o.g. that the parameter estimators of the
mixture components $H+1,\ldots,G$ leave a compact set $D$ 
of the form $D=[\pi_{min},1]\times C,\
C\subset \R^p \times {\cal S}_p$ compact, $\pi_{min}>0$.   
Then there exists $\phi_{min}$ bounding $\phi^+_j(x_i)$ from below for
$j=1,\ldots,H$ and $i=1,\ldots,n$, so 
$\sum_{j=1}^H\pi_j^+\phi_j^+(x_i)\ge H\pi_{min}\phi_{min}$.

Consider sequences $(\theta_m)_{m\in \N}\in \Theta$ with 
$l_{(n+r)G}(\underline{x_{n+r}},\theta_m)\to l_{(n+r)G}^o$ and leaving any $D$ for 
$j=H+1\ldots,G$, i.e., 
$\|\mu_{mj}\|\to\infty$ or $\lambda_{k,j,m}\to\infty$ or $\pi_{mj}\to 0$, but
with all $\lambda_{k,j,m}\ge\lambda_0$ as established in Corollary 
\ref{corollary:phimax}. Observe that for such sequences
$\sum_{j=H+1}^G \pi_{mj}\phi_{mj}(x_i)$ becomes arbitrarily small for
$i=1,\ldots,n$. Thus, for arbitrary $\epsilon>0$ and
$D$ large enough:
\begin{eqnarray*}
 l_{(n+r)G}^o & \le & \sum_{i=1}^n\log\left(\sum_{j=1}^H
  \pi_j^+\phi_j^+(x_i)+\pi_0^+\delta\right)+r\log \phi_{max}+\epsilon 
 \\
& \le & \max_{H<G} l_{nH}^o+r\log \phi_{max}+\epsilon
\le l_{n(G-1)}^o+r\log \phi_{max}+\epsilon.
\end{eqnarray*}
But a potential estimator $\hat\theta$ could be defined by
$\hat\pi_0=\frac{n\pi_0^*+r}{n+r},\ \hat\pi_j=\frac{n}{n+r}\pi_j^*,\
\hat \mu_j=\mu_j^*,\ \hat\Sigma_j=\Sigma_j^*,\ j=1,\ldots,G.$ 
Note that $\hat\theta\in \Theta$ because of \eqref{eq:breakdowncon2}.
Therefore,
\begin{eqnarray*}
 l_{(n+r)G}^o &\ge & \sum_{i=1}^n\log\left(\sum_{j=1}^G \pi_j^*
   \phi_j^*(x_i)+\left(\pi_0^*+\frac{r}{n}\right)\delta\right) \\ 
   && +r\log\left[\left(\pi_0^*+\frac{r}{n}\right)\delta\right]+(n+r)\log\frac{n}{n+r} \\
 \Rightarrow  
 l_{n(G-1)}^o & \ge & \sum_{i=1}^n\log\left(\sum_{j=1}^G \pi_j^*
   \phi_j^*(x_i)+\left(\pi_0^*+\frac{r}{n}\right)\delta\right) \\ 
   && +r\log\left[\left(\pi_0^*+\frac{r}{n}\right)\delta\right]+
(n+r)\log\frac{n}{n+r}-r\log
   \phi_{max}-\epsilon. 
\end{eqnarray*}
This contradicts \eqref{eq:breakdowncon1} by $\epsilon\to 0$.
\end{proof}



\CHANGE{

\section{Numerical experiments}\label{sec:experiment}

In this section, we perform Monte Carlo experiments to compare robust clustering methods on the two sampling designs introduced in Section \ref{sec_data}. There is already a comprehensive simulation study involving OTRIMLE and competitors in \cite{Coretto_Hennig_2014_comparison}, so here we use different setups. Note though that the algorithm \ref{algo:ecm} introduced here is different from the one used in \cite{Coretto_Hennig_2014_comparison} and in our view preferable. Below, apart from involving competing methods from the literature, the two OTRIMLE algorithms are compared.

The AsyNoise design of Figure \ref{fig:dataset_asynoise}  generates $G=5$ clusters in $p=20$ dimensions and an expected noise proportion of 33\%. 
The five clusters are generated from a mixture of  t-distributions with parameters given in Table \ref{tab:pars_asynoise}.  The five clusters show a combination of structures that are often difficult to handle together. Some of them are not well separated, they are of different size, and although they are all elliptically shaped, there are strong differences in cluster scatters, and deviations from normality. The noise originates from a distribution obtained as the product of two independent one-dimensional uniform distributions with support on the interval $[-25, 25]$, and $18$ independent one-dimensional  $\chi^2$-distributions with 1 degree of freedom. The first and the third marginal are distributed uniformly, producing background noise on both clustered and non-clustered dimensions, and the $\chi^2$-distribution adds a strong dose of asymmetry. 

\begin{table}[!t]
\centering
\caption{Parameters of the AsyNoise sampling design. Let $\pi$ and $\nu$ be the expected proportion and the degrees of freedom. $m_1,v_1$ and $m_2,v_2$ are the mean parameters ($m$) and variance parameters ($v$) along dimensions 1 and 2 respectively. $c_{1,2}$ denotes the covariance between marginals 1 and 2. All remaining variances are set equal to 1, while all remaining mean and covariance parameters are set equal to 0. }
\label{tab:pars_asynoise}
\begin{tabular}{lrrrrr}
	\toprule
	Parameter                                  &     \multicolumn{5}{c}{Cluster}      \\
	                                           &   1   &   2   &  3   &   4   &   5   \\
	\cmidrule{2-6}
    $\pi$                                    & 10.05\% & 20.10\% & 6.70\% & 10.05\% & 20.10\% \\
	$\nu$                                    &  10   &  11   &  12  &  13   &  14   \\
	$m_1$                                      &   0   &   7   &  5   &  -11  &  -7   \\
	$m_2$                                      &   3   &   1   &  9   &  11   &   5   \\
	$v_1$                                      &   1   &   2   &  2   &  0.5  &  2.5  \\
	$v_2$                                      &   1   &   2   &  2   &  0.5  &  2.5  \\
	$c_{12}$                                   &  0.5  & -1.5  & 1.3  &   0   &   0   \\
	\bottomrule
\end{tabular}
\end{table}
\begin{table}[!t]
\centering
\caption{Parameters of the GEM sampling design. Let $\pi$ be the expected proportion. $m$ is the mean parameter constant across all marginals for the same cluster. Each clusters has unit variance across all marginals and correlation matrix given by $C(\rho)$.}
\label{tab:pars_gem}
\begin{tabular}{lrr}
	\toprule
	Parameter                  & \multicolumn{2}{c}{Cluster} \\
	                           &      1 &                  2 \\
	\cmidrule{2-3}
    $\pi$                    & 29.4\% &             68.6\% \\
	$m$                      &      0 &                  4 \\
	$\rho$                   &   0.99 &                  0 \\
	\bottomrule
\end{tabular}
\end{table}

The second sampling design is referred to as GEM (see Figure \ref{fig:dataset_gem}), which stands for Gross Error Model. In this case, the sampling design is a mixture of two Gaussian distributions in $p=20$ dimensions, with the addition of a few potential outliers. In this design, the first cluster has strongly  correlated marginals, whereas the second one is spherical, and this produces a large discrepancy between the clusters' shapes. Define the $p{\times}p$ correlation matrix  $C(\rho):=(\rho^{|l-k|})_{l,k}$  for $l,k=1,\ldots, p$ (also called AR(1) correlation model). The parameters of the GEM design are specified in Table \ref{tab:pars_gem}.  An expected $2\%$ of points are generated from a 20-dimensional t-distribution with 3 degrees of freedom, centered at  $(0,0, -7, \ldots, -7)^{\prime}$, with unit variances and correlation matrix $C(0.9999)$. This produces a few points far from both clusters, although these outliers are not extremely separated from the regular data. While non-robust methods can cope with weakly separated outliers at the expense of large estimation bias, some robust methods capable of handling extreme outliers might get in trouble if the separation gap between regular and nonregular points is modest. 

In this experiment ML for Gaussian mixtures with uniform noise and TCLUST are compared with the RIMLE optimally tuned according to the OTRIMLE method introduced in \cite{Coretto_Hennig_2014_comparison}. Methods under comparison are set up as follow:
\begin{description}
\item[OtrimleECM:] RIMLE is computed based on the ECM algorithm \ref{algo:ecm} on a grid of 50 $\log(\delta))$ values as described in Section \ref{sec_choice_of_delta}. The OTRIMLE criterion proposed in \cite{Coretto_Hennig_2014_comparison} selects the best solution. 
The input parameter $\pimax$ is always set to the conventional 50\%. The eigenratio constraint is varied between the strongest restriction ($\gamma=0$), and no restriction at all ($\gamma=+\infty$). In particular, $\log_{10}{(\gamma)}=\{0, 0.5, 1, 2, 3, 6, +\infty$\}. The initial partition is computed as described in Section \ref{sec_choice_of_delta}. OtrimleECM is computed using the \pkg{otrimle} package of \cite{CRAN_otrimle}.

\item[OtrimleAEM:] RIMLE is computed using the approximate EM-algorithm introduced in \cite{Coretto_Hennig_2014_comparison}. Both $\pimax$ and $\gamma$ are set for OtrimleECM as well as initial values. Software for OtrimleAEM is available as part of the supplementary materials in \cite{Coretto_Hennig_2014_comparison}. 

\item[TclustOracle:] TCLUST with trimming rate set to the true underlying noise proportion. Eigenratio constraint is also treated as for OtrimleECM. TclustOracle is computed using the \pkg{tclust} package of \cite{Fritz_etal_2012} which does not allow the user to choose an initial partition. TCLUST initialization is random, and we increased the default number of random starts to the sample size. Default maximum number of iterations is also increased to 500 because several convergence problems were recorded. 

\item[TclustFix:]same as TclustOracle but with trimming rate fixed to a low 5\% for the GEM design, and a high 50\% for the AsyNoise design. We tried to run the ctlcurves-tool of the \pkg{tcluts}  package on several data sets without a clear indication.  

\item[MCLUST:] ML for Gaussian mixtures with uniform noise as implemented in the \pkg{mclust} package of \cite{Mclust_Software_Manual_2012}. Regularization of the covariance matrices is done by choosing an appropriate covariance parameterization based on the BIC (Bayesian Information Criterion). Mclust requires noise initialization, and this is initialized as for the OtrimleECM. Note that the \pkg{otrimle} package uses \pkg{mclust} initialization for the regular points, hence OtrimleECM and MCLUST both start from the same partition. 
\end{description}
Although we considered DBSCAN in Section \ref{sec_data}, it is not considered here because its results strongly depend on a pair of interdependent tunings which needs to be carefully selected based on data. 

Sample size is set to $n=500$ for AsyNoise, and $n=100$ for GEM. With these relatively low sample size, the regularization of the covariance matrices becomes crucial because often small clusters are found compared to the dimensionality $p$.   For both data sets 1000 Monte Carlo replicates have been considered.  The true cluster label of a point is defined based on the component of the sampling distribution that generates it. Misclassification rates are computed with respect to the minimizing permutation of clusters' indexes not involving the estimated noise, which is always matched to the true noise. The underlying eigenratio behavior of this designs is largely varying. The true $\gamma$ is 7 for AsyNoise, and it is 3704.7 for GEM. However,  if one computes the eigenratio of sample clusters' covariances based on true labels the figure can be completely different. In fact, we computed the (5\%, 95\%)-quantiles of the  Monte Carlo distribution of these quantities, and we obtained (44.5, 273.4) for AsyNoise, and (19899.3, 246826.8) for GEM. In the examples given in Section \ref{sec_data} we fixed $\gamma=100$, because in real world applications one typically does not have information on it, and we used the median value adopted in these experiments. 

\begin{table}[!t]
\centering
\caption{Monte Carlo averages, with standard errors in parenthesis,  of misclassification rates (\%) for the AsyNoise sampling design. ``na'' is reported if the software did not produce a valid answer in more than 50\% of the replicates.}
\label{tab:mcr_asynoise}
\begin{tabular}{rrrrrr}
	\toprule
	$\log_{10}(\gamma)$ &  OtrimleECM &  OtrimleAEM & TclustOracle &   TclustFix &      MCLUST \\
	\midrule
	                  0 & 15.31(0.00) & 31.96(0.03) &   4.33(0.00) & 18.63(0.00) &         --- \\
	                0.5 & 11.25(0.01) & 25.57(0.03) &   5.31(0.00) & 18.65(0.00) &         --- \\
	                  1 &  9.46(0.00) & 22.14(0.02) &  16.52(0.00) & 23.30(0.00) &         --- \\
	                  2 & 11.48(0.00) & 19.40(0.01) &  50.26(0.00) & 48.66(0.00) &         --- \\
	                  3 & 12.37(0.01) & 19.71(0.01) &  57.88(0.00) & 56.72(0.00) &         --- \\
	                  6 & 12.05(0.01) & 19.90(0.01) &  57.26(0.00) & 56.89(0.00) &         --- \\
	          $+\infty$ & 12.08(0.00) & 19.92(0.01) &           na &          na & 27.05(0.02) \\
	\bottomrule
\end{tabular}
\end{table}

\begin{table}[!t]
\centering
\caption{Monte Carlo averages, with standard errors in parenthesis,  of misclassification rates (\%) for the GEM sampling design. ``na'' is reported if the software did not produce a valid answer in more than 50\% of the replicates.}
\label{tab:mcr_gem}
\begin{tabular}{rrrrrr}
	\toprule
	$\log_{10}(\gamma)$ & OtrimleECM &  OtrimleAEM & TclustOracle &   TclustFix &      MCLUST \\
	\midrule
	                  0 & 1.10(0.00) & 63.97(0.04) &   1.78(0.00) &  3.32(0.00) &         --- \\
	                0.5 & 0.87(0.00) & 11.96(0.03) &   1.50(0.00) &  3.17(0.00) &         --- \\
	                  1 & 1.57(0.01) &  3.49(0.01) &   1.25(0.00) &  3.11(0.00) &         --- \\
	                  2 & 0.52(0.00) &  2.25(0.00) &   8.10(0.01) &  9.68(0.01) &         --- \\
	                  3 & 0.50(0.00) &  0.71(0.00) &  16.41(0.00) & 18.08(0.00) &         --- \\
	                  6 & 3.82(0.01) &  1.18(0.01) &  14.07(0.00) & 16.80(0.00) &         --- \\
	          $+\infty$ & 4.66(0.01) &  1.17(0.01) &           na &          na & 41.60(0.01) \\
	\bottomrule
\end{tabular}
\end{table}

\begin{figure}[!t]
\centering
\includegraphics[width=\linewidth]{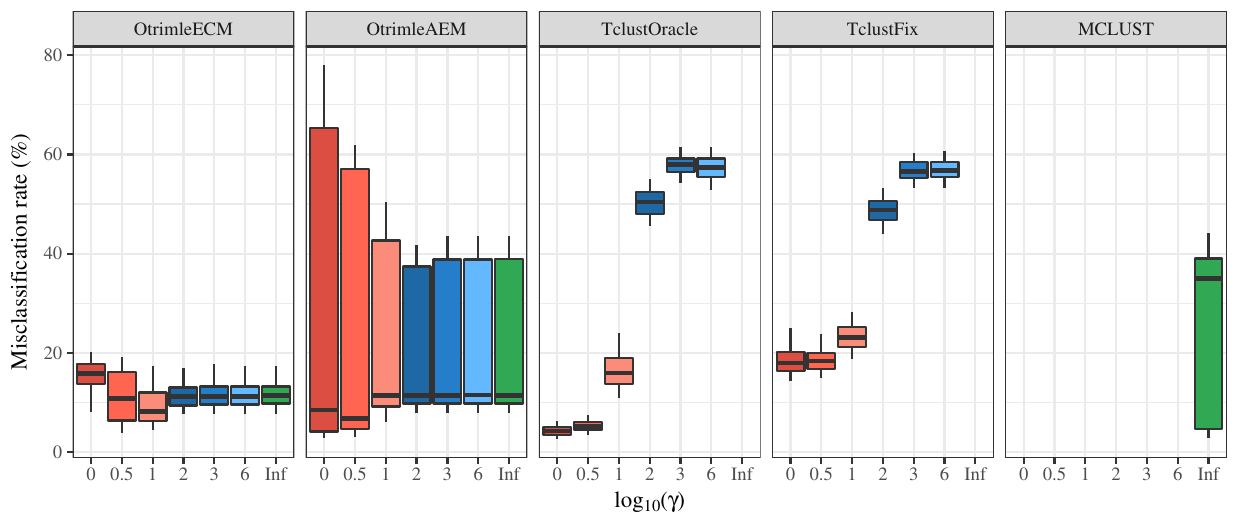}
\caption{Modified boxplots of the Monte Carlo distribution of misclassification rates for the AsyNoise design: whiskers coincide with (5\%, 95\%)--quantiles of the distribution.}
\label{fig:boxplot_AsyNoise}
\end{figure}

\begin{figure}[!t]
\centering
\includegraphics[width=\linewidth]{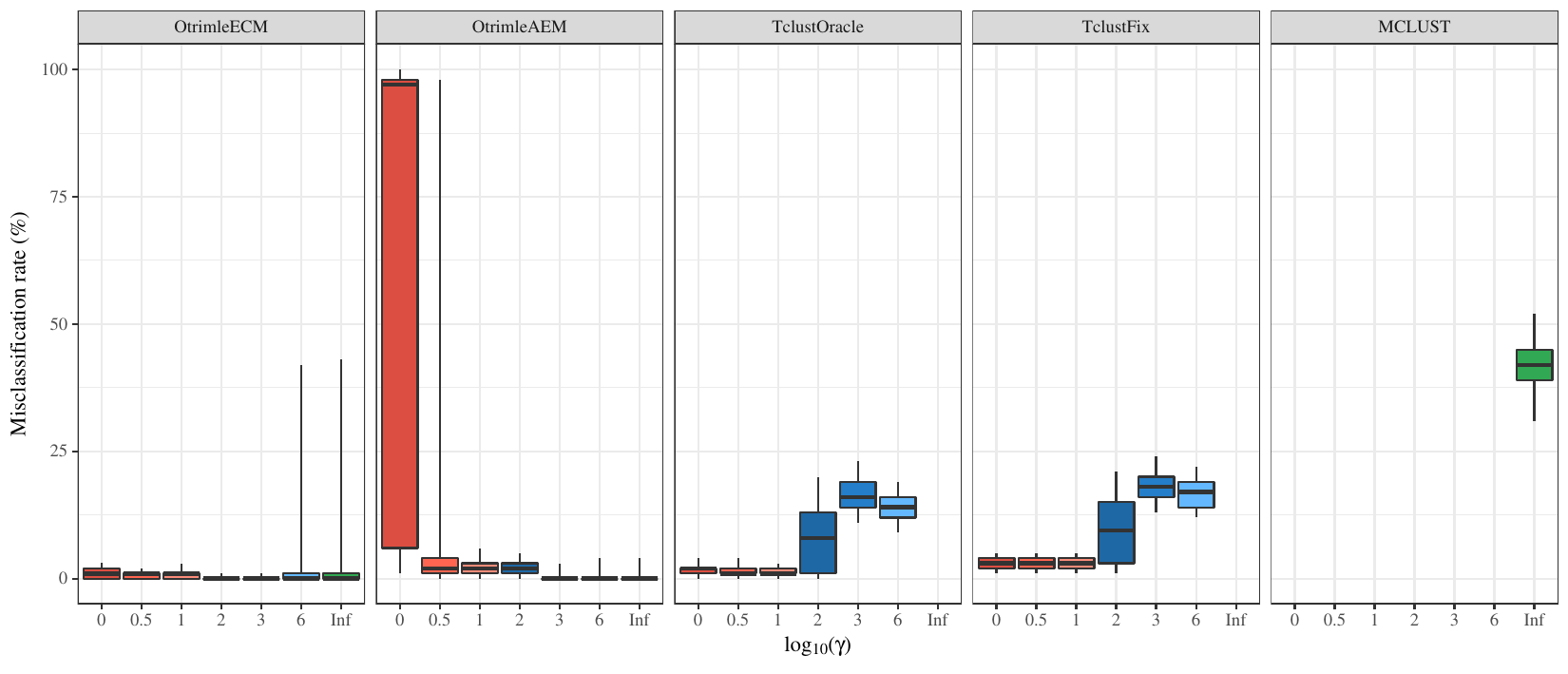}
\caption{Modified boxplots of the Monte Carlo distribution of misclassification rates for the GEM design: whiskers coincide with (5\%, 95\%)--quantiles of the distribution.}
\label{fig:boxplot_GEM}
\end{figure}

Results are summarized in Tables \ref{tab:mcr_asynoise} and \ref{tab:mcr_gem}, and Figures \ref{fig:boxplot_AsyNoise} and \ref{fig:boxplot_GEM}. Since MCLUST does not enforce an eigenratio constraint, results are recorded at  $\gamma=+\infty$, although MCLUST has its own covariance regularization. MCLUST is seriously affected by contamination in both designs, its performance is better for the AsyNoise design for which the boxplot of Figure \ref{fig:boxplot_AsyNoise} shows that in some replica it can produce misclassification rates below 10\%. Regarding the Otrimle and Tclust versions, the performance depends on the setting of the eigenratio constraint. However, OtrimleECM offers the most stable performance in both designs.  OtrimleECM achieves the best misclassification performance in all situations except for few cases where TclustOracle does better, but in fact, TclustOracle is run with the assumption that one knows exactly the expected amount of noise, which is never true in reality.  It is worth to note that TclustOracle seems to not tolerate large values of $\gamma$ in both designs.  This is counterintuitive at least for GEM, where the true $\log_{10}(\gamma)$ is between 3 and 6, but  in this range both TclustOracle and TclustFix have serious problems. OtrimleAEM is  the second best overall, although it shows a large positive skewness in the distribution of the misclassification rates for AsyNoise, and in both designs it is not able to enforce low $\gamma$ values appropriately. This is because in OtrimleAEM an approximate eigenratio constraint is applied at the end of the EM iteration, while during the iteration only a minimum determinant condition is controlled. These results show a remarkable improvement of the ECM algorithm \ref{algo:ecm} over the approximate solution proposed in \cite{Coretto_Hennig_2014_comparison}.


In practice the user has to specify $\gamma$. According to the results shown here, OtrimleECM is not very sensitive to this choice. Also the results show that good misclassification rates can be achieved in GEM, with a true $\gamma>3000$, using a much lower $\gamma$ for ORIMLE; actually for TCLUST a much lower $\gamma$ is even required to achieve good results. Choosing a lower $\gamma$ in such situations may provide some welcome regularization. $\gamma=100$ often seems to be a sensible choice. However, the user needs to have in mind that a straight interpretation of $\gamma$ requires that a variance of 1 (say) along a one-dimensional projection has the same meaning in all directions in data space, which is particularly doubtful if variables have different measurement units or variable-wise variations are not meaningfully comparable. In such cases sphering or at least variable standardization may be advisable.

}

\section{Concluding Remarks}\label{sec_concluding_remarks}

The RIMLE robustifies the MLE in the Gaussian mixture model by adding  an improper constant mixture component to catch outliers and points that cannot appropriately assigned to any cluster. Characteristics of the method compared to other robust clustering methods aiming for approximately Gaussian clusters are a smooth mixture-type transition between clusters and noise, and the fact  that noise and outliers are not modelled by a specific and usually misspecified distribution, but rather as anything where the estimated mixture density is so low that the observation is rather classified to the constant noise than to any mixture component. If needed, the density value of the  improper constant noise component can be chosen in a data-adaptive based on the OTRIMLE criterion developed in   \citep{Coretto_Hennig_2014_comparison}. The RIMLE/OTRIMLE  has shown competitive performance when compared with state of the art methods for robust model-based clustering methods. In this paper we investigated theoretical properties of the  RIMLE, and it is shown existence, consistency, breakdown behaviour, and convergence of algorithms.  Since the RIMLE coincides with the MLE for Gaussian finite mixture models (when $\delta=0$), the present paper also  gives a comprehensive treatment for it which was  missing in the literature.


\bibliographystyle{chicago}
\bibliography{REF}%

\begin{thebibliography}{}

\bibitem[\protect\citeauthoryear{Alexandrovich}{Alexandrovich}{2014}]{Alexandrovich_2014}
Alexandrovich, G. (2014).
\newblock A note on the article `{I}nference for multivariate normal mixtures'
  by {J}. {C}hen and {X}. {T}an.
\newblock {\em J. Multivariate Anal.\/}~{\em 129}, 245--248.

\bibitem[\protect\citeauthoryear{Anderson and Olkin}{Anderson and
  Olkin}{1985}]{Anderson_Olkin_1985}
Anderson, T.~W. and I.~Olkin (1985).
\newblock Maximum-likelihood estimation of the parameters of a multivariate
  normal distribution.
\newblock {\em Linear Algebra Appl.\/}~{\em 70}, 147--171.

\bibitem[\protect\citeauthoryear{Banfield and Raftery}{Banfield and
  Raftery}{1993}]{Banfield_Raftery_1993}
Banfield, J.~D. and A.~E. Raftery (1993).
\newblock Model-based gaussian and non-gaussian clustering.
\newblock {\em Biometrics\/}~{\em 49}, 803--821.

\bibitem[\protect\citeauthoryear{Bertsekas}{Bertsekas}{1999}]{Bertsekas_1999}
Bertsekas, D.~P. (1999).
\newblock {\em Nonlinear Programming}.
\newblock Athena Scientific.

\bibitem[\protect\citeauthoryear{Chen and Tan}{Chen and
  Tan}{2009}]{Chen_Tan_2009}
Chen, J. and X.~Tan (2009).
\newblock Inference for multivariate normal mixtures.
\newblock {\em J. Multivariate Anal.\/}~{\em 100\/}(7), 1367--1383.

\bibitem[\protect\citeauthoryear{Coretto and Hennig}{Coretto and
  Hennig}{2010}]{Coretto_Hennig_2010}
Coretto, P. and C.~Hennig (2010).
\newblock A simulation study to compare robust clustering methods based on
  mixtures.
\newblock {\em Advances in Data Analysis and Classification\/}~{\em 4\/}(2),
  111--135.

\bibitem[\protect\citeauthoryear{Coretto and Hennig}{Coretto and
  Hennig}{2011}]{Coretto_Hennig_2011}
Coretto, P. and C.~Hennig (2011).
\newblock Maximum likelihood estimation of heterogeneous mixtures of gaussian
  and uniform distributions.
\newblock {\em Journal of Statistical Planning and Inference\/}~{\em 141\/}(1),
  462--473.

\bibitem[\protect\citeauthoryear{Coretto and Hennig}{Coretto and
  Hennig}{2016}]{Coretto_Hennig_2014_comparison}
Coretto, P. and C.~Hennig (2016).
\newblock {R}obust improper maximum likelihood: tuning, computation, and a
  comparison with other methods for robust {G}aussian clustering.
\newblock {\em Journal of the American Statistical Association\/}~{\em 111},
  1648--1659.

\bibitem[\protect\citeauthoryear{Coretto and Hennig}{Coretto and
  Hennig}{2017}]{CRAN_otrimle}
Coretto, P. and C.~Hennig (2017).
\newblock otrimle: Robust model-based clustering.
\newblock R package version 1.0. Available at:
  https://CRAN.R-project.org/package=otrimle.

\bibitem[\protect\citeauthoryear{Cuesta-Albertos, Gordaliza, and
  Matr\'an}{Cuesta-Albertos et~al.}{1997}]{Cuesta_Albertos_etal_1997}
Cuesta-Albertos, J.~A., A.~Gordaliza, and C.~Matr\'an (1997).
\newblock Trimmed k-means: An attempt to robustify quantizers.
\newblock {\em Annals of Statistics\/}~{\em 25}, 553--576.

\bibitem[\protect\citeauthoryear{Dennis}{Dennis}{1981}]{Dennis_1981}
Dennis, J. E.~J. (Ed.) (1981).
\newblock {\em Algorithms for nonlinear fitting}, Cambridge, England. NATO
  advanced Research Symposium: Cambridge University Press.

\bibitem[\protect\citeauthoryear{Ester, Kriegel, Sander, and Xu}{Ester
  et~al.}{1996}]{EsKrSaXu_1996}
Ester, M., H.-P. Kriegel, J.~Sander, and X.~Xu (1996).
\newblock A density-based algorithm for discovering clusters in large spatial
  databases with noise.
\newblock In {\em Proceedings of the 2nd International Conference on Knowledge
  Discovery and Data Mining (KDD-96)}, pp.\  226--231. Institute for Computer
  Science, University of Munich.

\bibitem[\protect\citeauthoryear{Fraley}{Fraley}{1998}]{Fraley1998}
Fraley, C. (1998, jan).
\newblock Algorithms for model-based gaussian hierarchical clustering.
\newblock {\em {SIAM} Journal on Scientific Computing\/}~{\em 20\/}(1),
  270--281.

\bibitem[\protect\citeauthoryear{Fraley and Raftery}{Fraley and
  Raftery}{2002}]{Fraley_Raftery_2002}
Fraley, C. and A.~E. Raftery (2002).
\newblock Model-based clustering, discriminant analysis, and density
  estimation.
\newblock {\em Journal of the American Statistical Association\/}~{\em 97},
  611--631.

\bibitem[\protect\citeauthoryear{Fraley, Raftery, Murphy, and Scrucca}{Fraley
  et~al.}{2012}]{Mclust_Software_Manual_2012}
Fraley, C., A.~E. Raftery, T.~B. Murphy, and L.~Scrucca (2012).
\newblock mclust version 4 for r: Normal mixture modeling for model-based
  clustering, classification, and density estimation.
\newblock Technical Report 597, University of Washington, Department of
  Statistics.

\bibitem[\protect\citeauthoryear{Fritz, Garc\'{\i}a-Escudero, and
  Mayo-Iscar}{Fritz et~al.}{2012}]{Fritz_etal_2012}
Fritz, H., L.~A. Garc\'{\i}a-Escudero, and A.~Mayo-Iscar (2012).
\newblock {tclust}: An {R} package for a trimming approach to cluster analysis.
\newblock {\em Journal of Statistical Software\/}~{\em 47\/}(12), 1--26.

\bibitem[\protect\citeauthoryear{Fritz, Garc{\'{\i}}a-Escudero, and
  Mayo-Iscar}{Fritz et~al.}{2013}]{Fritz_etal_2013}
Fritz, H., L.~A. Garc{\'{\i}}a-Escudero, and A.~Mayo-Iscar (2013).
\newblock A fast algorithm for robust constrained clustering.
\newblock {\em Comput. Statist. Data Anal.\/}~{\em 61}, 124--136.

\bibitem[\protect\citeauthoryear{Gallegos}{Gallegos}{2002}]{Gallegos_2002}
Gallegos, M.~T. (2002).
\newblock Maximum likelihood clustering with outliers.
\newblock In {\em Classification, Clustering, and Data Analysis}, pp.\
  247--255. Springer.

\bibitem[\protect\citeauthoryear{Gallegos and Ritter}{Gallegos and
  Ritter}{2005}]{Gallegos_Ritter_2005}
Gallegos, M.~T. and G.~Ritter (2005).
\newblock A robust method for cluster analysis.
\newblock {\em Annals of Statistics\/}~{\em 33\/}(5), 347--380.

\bibitem[\protect\citeauthoryear{Gallegos and Ritter}{Gallegos and
  Ritter}{2009}]{Gallegos_Ritter_2009}
Gallegos, M.~T. and G.~Ritter (2009).
\newblock Trimmed {ML} estimation of contaminated mixtures.
\newblock {\em Sankhya (Ser. A)\/}~{\em 71}, 164--220.

\bibitem[\protect\citeauthoryear{Gallegos and Ritter}{Gallegos and
  Ritter}{2013}]{Gallegos_Ritter_2013}
Gallegos, M.~T. and G.~Ritter (2013).
\newblock Strong consistency of $k$-parameters clustering.
\newblock {\em Journal of Multivariate Analysis\/}~{\em 117}, 14--31.

\bibitem[\protect\citeauthoryear{Garc\'{\i}a-Escudero and
  Gordaliza}{Garc\'{\i}a-Escudero and
  Gordaliza}{1999}]{GarciaEscudero_Gordaliza_1999}
Garc\'{\i}a-Escudero, L.~A. and A.~Gordaliza (1999).
\newblock Robustness properties of $k$-means and trimmed $k$-means.
\newblock {\em Journal of the American Statistical Association\/}~{\em 94},
  956--969.

\bibitem[\protect\citeauthoryear{Garc\'ia-Escudero, Gordaliza, Matr\'an, and
  Mayo-Iscar}{Garc\'ia-Escudero et~al.}{2008}]{GarciaEscudero_etal_2008}
Garc\'ia-Escudero, L.~A., A.~Gordaliza, C.~Matr\'an, and A.~Mayo-Iscar (2008).
\newblock A general trimming approach to robust cluster analysis.
\newblock {\em Annals of Statistics\/}~{\em 38\/}(3), 1324--1345.

\bibitem[\protect\citeauthoryear{Garc{\'\i}a-Escudero, Gordaliza, Matr{\'a}n,
  and Mayo-Iscar}{Garc{\'\i}a-Escudero et~al.}{2014}]{GarciaEscudero_etal_2014}
Garc{\'\i}a-Escudero, L.~A., A.~Gordaliza, C.~Matr{\'a}n, and A.~Mayo-Iscar
  (2014).
\newblock Avoiding spurious local maximizers in mixture modeling.
\newblock {\em Statistics and Computing\/}~{\em 25}, 1--15.

\bibitem[\protect\citeauthoryear{Garc{\'\i}a-Escudero, Gordaliza, Matr{\'a}n,
  Mayo-Iscar, and Hennig}{Garc{\'\i}a-Escudero
  et~al.}{2015}]{GarciaEscudero_etal_2015}
Garc{\'\i}a-Escudero, L.~A., A.~Gordaliza, C.~Matr{\'a}n, A.~Mayo-Iscar, and
  C.~Hennig (2015).
\newblock Robustness and outliers.
\newblock In C.~Hennig, M.~Meila, F.~Murtagh, and R.~Rocci (Eds.), {\em
  Handbook of Cluster Analysis}, pp.\  653--678. Boca Raton FL: CRC Press.

\bibitem[\protect\citeauthoryear{Gavish and Donoho}{Gavish and
  Donoho}{2017}]{Gavish_Donoho_2014}
Gavish, M. and D.~L. Donoho (2017).
\newblock Optimal shrinkage of singular values.
\newblock {\em {IEEE} Transactions on Information Theory\/}~{\em 63\/}(4),
  2137--2152.

\bibitem[\protect\citeauthoryear{Hahsler}{Hahsler}{2016}]{R_dbscan}
Hahsler, M. (2016).
\newblock {\em dbscan: Density Based Clustering of Applications with Noise
  (DBSCAN) and Related Algorithms}.
\newblock R package version 0.9-7.

\bibitem[\protect\citeauthoryear{Hathaway}{Hathaway}{1985}]{Hathaway_1985}
Hathaway, R.~J. (1985).
\newblock A constrained formulation of maximum-likelihood estimation for normal
  mixture distributions.
\newblock {\em The Annals of Statistics\/}~{\em 13}, 795--800.

\bibitem[\protect\citeauthoryear{Hennig}{Hennig}{2004}]{Hennig_2004}
Hennig, C. (2004).
\newblock Breakdown points for maximum likelihood estimators of location-scale
  mixtures.
\newblock {\em The Annals of Statistics\/}~{\em 32\/}(4), 1313--1340.

\bibitem[\protect\citeauthoryear{Hennig}{Hennig}{2008}]{Hennig_2008}
Hennig, C. (2008).
\newblock Dissolution point and isolation robustness: robustness criteria for
  general cluster analysis methods.
\newblock {\em Journal of Multivariate Analysis\/}~{\em 99}, 1154--1176.

\bibitem[\protect\citeauthoryear{Hennig}{Hennig}{2015a}]{Hennig2015b}
Hennig, C. (2015a).
\newblock Clustering strategy and method selection.
\newblock In C.~Hennig, M.~Meila, F.~Murtagh, and R.~Rocci (Eds.), {\em
  Handbook of Cluster Analysis}, Chapter~31, pp.\  703--730. Chapman \&
  Hall/CRC, Boca Raton FL.

\bibitem[\protect\citeauthoryear{Hennig}{Hennig}{2015b}]{Hennig_2015}
Hennig, C. (2015b).
\newblock What are the true clusters?
\newblock {\em Pattern Recognition Letters\/}~{\em 64}, 53--62.

\bibitem[\protect\citeauthoryear{Hennig and Liao}{Hennig and
  Liao}{2013}]{Hennig_Liao_2013}
Hennig, C. and T.~F. Liao (2013).
\newblock How to find an appropriate clustering for mixed type variables with
  application to socioeconomic stratification (with discussion).
\newblock {\em Journal of the Royal Statistical Science, Series C (Applied
  Statistics)\/}~{\em 62}, 309--369.

\bibitem[\protect\citeauthoryear{Ingrassia}{Ingrassia}{2004}]{Ingrassia_2004}
Ingrassia, S. (2004).
\newblock A likelihood-based constrained algorithm for multivariate normal
  mixture models.
\newblock {\em Statistical Methods \& Applications\/}~{\em 13\/}(2), 151--166.

\bibitem[\protect\citeauthoryear{Ingrassia and Rocci}{Ingrassia and
  Rocci}{2007}]{Ingrassia_Rocci_2007}
Ingrassia, S. and R.~Rocci (2007).
\newblock Constrained monotone {EM} algorithms for finite mixture of
  multivariate gaussians.
\newblock {\em Computational Statistics and Data Analysis\/}~{\em 51},
  5339--5351.

\bibitem[\protect\citeauthoryear{Ingrassia and Rocci}{Ingrassia and
  Rocci}{2011}]{Ingrassia_Rocci_2011}
Ingrassia, S. and R.~Rocci (2011).
\newblock Degeneracy of the {EM} algorithm for the {MLE} of multivariate
  gaussian mixtures and dynamic constraints.
\newblock {\em Computational statistics \& data analysis\/}~{\em 55\/}(4),
  1715--1725.

\bibitem[\protect\citeauthoryear{Jennrich}{Jennrich}{1969}]{Jennrich_1969}
Jennrich, R.~I. (1969).
\newblock Asymptotic properties of non-linear least squares estimators.
\newblock {\em Annals of Mathematical Statistics\/}~{\em 40}, 633--643.

\bibitem[\protect\citeauthoryear{Kiefer}{Kiefer}{1953}]{Kiefer_1953}
Kiefer, J. (1953).
\newblock Sequential minimax search for a maximum.
\newblock {\em Proceedings of the American Mathematical Society\/}~{\em
  4\/}(3), 502--506.

\bibitem[\protect\citeauthoryear{Kiefer and Wolfowitz}{Kiefer and
  Wolfowitz}{1956}]{Kiefer_Wolfowitz_1956}
Kiefer, N.~M. and J.~Wolfowitz (1956).
\newblock Consistency of the maximum likelihood estimation in the presence of
  infinitely many incidental parameter.
\newblock {\em Annals of Mathematical Statistics\/}~{\em 27\/}(364), 887--906.

\bibitem[\protect\citeauthoryear{McLachlan and Peel}{McLachlan and
  Peel}{2000}]{McLachlan_Peel_2000b}
McLachlan, G.~J. and D.~Peel (2000).
\newblock Robust mixture modelling using the t--distribution.
\newblock {\em Statistics and Computing\/}~{\em 10\/}(4), 339--348.

\bibitem[\protect\citeauthoryear{Meng and Rubin}{Meng and
  Rubin}{1993}]{Meng_Rubin_1993}
Meng, X.-L. and D.~B. Rubin (1993).
\newblock Maximum likelihood estimation via the {ECM} algorithm: a general
  framework.
\newblock {\em Biometrika\/}~{\em 80\/}(2), 267--278.

\bibitem[\protect\citeauthoryear{Redner}{Redner}{1981}]{Redner_1981}
Redner, R. (1981).
\newblock Note on the consistency of the maximum likelihood estimate for
  nonidentifiable distributions.
\newblock {\em The Annals of Statistics\/}~{\em 9}, 225--228.

\bibitem[\protect\citeauthoryear{Redner and Walker}{Redner and
  Walker}{1984}]{Redner_Walker_1984}
Redner, R.~A. and H.~F. Walker (1984).
\newblock Mixture densities, maximum likelihood and the {EM} algorithm.
\newblock {\em SIAM Review\/}~{\em 26}, 195--239.

\bibitem[\protect\citeauthoryear{Ritter}{Ritter}{2014}]{Ritter_2014}
Ritter, G. (2014).
\newblock {\em Robust Cluster Analysis and Variable Selection}.
\newblock Monographs on Statistics and Applied Probability. Chapman and
  Hall/CRC.

\bibitem[\protect\citeauthoryear{Theobald}{Theobald}{1975}]{Theobald_1975}
Theobald, C. (1975).
\newblock An inequality with application to multivariate analysis.
\newblock {\em Biometrika\/}~{\em 62\/}(2), 461--466.

\bibitem[\protect\citeauthoryear{{Van der Vaart}}{{Van der
  Vaart}}{2000}]{Vaart_2000}
{Van der Vaart}, A.~W. (2000).
\newblock {\em Asymptotic statistics}, Volume~3.
\newblock Cambridge university press.

\bibitem[\protect\citeauthoryear{van~der Vaart and Wellner}{van~der Vaart and
  Wellner}{1996}]{Van_der_Vaart_Wellner_1996}
van~der Vaart, A.~W. and J.~A. Wellner (1996).
\newblock {\em Weak Convergence and Empirical Processes}.
\newblock New York: Springer.

\bibitem[\protect\citeauthoryear{Won, Lim, Kim, and Rajaratnam}{Won
  et~al.}{2013}]{Won_Lim_Kim_Rajaratnam_2013}
Won, J.-H., J.~Lim, S.-J. Kim, and B.~Rajaratnam (2013).
\newblock Condition-number-regularized covariance estimation.
\newblock {\em J. R. Stat. Soc. Ser. B. Stat. Methodol.\/}~{\em 75\/}(3),
  427--450.

\bibitem[\protect\citeauthoryear{Wu}{Wu}{1983}]{Wu_1983}
Wu, C. F.~J. (1983).
\newblock On the convergence properties of the {EM} algorithm.
\newblock {\em Annals of Statistics\/}~{\em 11\/}(1), 95--103.

\end{thebibliography}
\end{document}